\documentclass[a4paper,UKenglish,cleveref, autoref, thm-restate]{lipics-v2021}

\hideLIPIcs  

\usepackage{hyperref,color}
\usepackage{todonotes}
\usepackage{amsmath,amssymb,amsthm}

\usepackage{algorithm}
\usepackage[noend]{algpseudocode}%

\usepackage[most]{tcolorbox}

\usepackage{graphicx}

\usepackage{lineno}
\nolinenumbers

\newcommand{\localSize}{s}

\newcommand{\numM}{m}
\newcommand{\setM}{\mathcal M}

\newcommand{\corrM}[1]{M_{#1}}

\newcommand{\metric}{\mathcal{X}}
\newcommand{\ddim}{\textsf{ddim}}
\newcommand{\dist}[1]{\textsf{dist}_{#1}}
\newcommand{\diam}{\textsf{diam}}

\newcommand{\cone}{\mathcal C}

\newcommand{\uVec}{\textbf{u}}

\newcommand{\grid}[1]{G_{#1}}

\newcommand{\cell}{\square}

\newcommand{\qt}{T}    
\newcommand{\lvl}[1]{\ell_{#1}}
\newcommand{\parent}[1]{\textsf{par}(#1)}
\newcommand{\ithDepParent}[2]{\textsf{par}^{(#2)}(#1)}

\newcommand{\nodeCell}[1]{\cell_{#1}}

\newcommand{\wspd}{\mathcal W}

\newcommand{\ccheck}[1]{\textcolor{black}{#1}}
\newcommand{\prob}{\text{Pr}}
\newcommand{\eps}{\varepsilon}

\newcommand{\primNode}[1]{v_{#1, \uVec}}

\newtheorem{property}[theorem]{Property}

\title{Fully Scalable MPC Algorithms for WSPD in Doubling and Euclidean Spaces} 

\author{Eunjin Oh}{Department of Computer Science and Engineering, POSTECH, Pohang, Korea}{eunjin.oh@postech.ac.kr}{0000-0003-0798-2580}{}
\author{Hyeonjun Shin}{Department of Computer Science and Engineering, POSTECH, Pohang, Korea}{hyeonjun.shin@postech.ac.kr}{0009-0008-4701-7295}{} 

\authorrunning{E Oh, and H Shin} 

\Copyright{Eunjin Oh, and Hyeonjun Shin} 



\category{} 

\funding{Supported by Institute of Information \& Communications Technology Planning \& Evaluation (IITP) grant funded by the Korea government (MSIT) (No. RS-2024-00440239, Sublinear Scalable Algorithms for Large-Scale Data Analysis) and the National Research Foundation of Korea (NRF) grant funded by the Korea government (MSIT) (No. RS-2024-00358505).}

\EventEditors{John Q. Open and Joan R. Access}
\EventNoEds{2}
\EventLongTitle{42nd Conference on Very Important Topics (CVIT 2016)}
\EventShortTitle{CVIT 2016}
\EventAcronym{CVIT}
\EventYear{2016}
\EventDate{December 24--27, 2016}
\EventLocation{Little Whinging, United Kingdom}
\EventLogo{}
\SeriesVolume{42}
\ArticleNo{23}

\begin{document}

\maketitle

\begin{abstract}
In this paper, we study the problem of constructing a $(1/\varepsilon)$-\emph{well-separated pair decomposition} (WSPD) for a point set of size $n$ in the Massively Parallel Computation (MPC) model, where multiple machines work in parallel and communicate in synchronous rounds.
We present an $O(1)$-round MPC algorithm that constructs a $O(1/\varepsilon)$-WSPD of size $(1/\varepsilon)^{O(\ddim)}\cdot \tilde O(n)$ for point sets in a metric space of a constant doubling dimension $\ddim$, with high probability, using $(1/\varepsilon)^{O(\ddim)} \cdot \tilde O(n)$ total space and $O(n^\delta)$ space per machine for a constant $\delta\in (0,1)$.
In the $d$-dimensional Euclidean space, we can improve the size of the WSPD and the total space to $(1/\varepsilon)^{O(d)} n$.
This improves the best-known algorithm [FOCS'93] for computing a WSPD which requires $O(\log n)$ rounds and works only in Euclidean spaces.
As a consequence, the following problems can be solved in $O(1)$ rounds in the MPC model: computing a $(1+\varepsilon)$-spanner, a $(1-\varepsilon)$-approximation of the diameter, the closest pair, and the $k$-nearest neighbors ($k$-NN).
While our $k$-NN algorithm is specific to Euclidean space, the other three problems can be solved in both Euclidean and doubling metric spaces.
\end{abstract}

\section{Introduction} \label{sec:intro}
In this paper, we study geometric proximity problems in the \emph{Massively Parallel Computation} (MPC) model, introduced by~\cite{karloff2010model}, now a standard framework for large-scale data processing in distributed systems such as MapReduce~\cite{dean2008mapreduce}, Dryad~\cite{isard2007dryad}, Hadoop~\cite{white2012hadoop}, and Spark~\cite{zaharia2010spark}. 
Here, the input data is distributed across multiple machines, and computation proceeds in synchronous rounds of local computation and communication. 
A central algorithmic goal is to minimize the number of communication rounds and to use small local memory while keeping the total work nearly linear.
Over the past decade, significant progress has been made on fundamental MPC problems such as sorting, connectivity, and matching~\cite{andoni2018parallel, assadi2019coresets, balliu2023optimal,chang2024fully, coy2022deterministic, dory2021constant,ghaffari2018improved,goodrich2011sorting,lahn2023combinatorial,nowicki2021dynamic}. 
While some geometric problems such as clustering, Euclidean MST, and spanners have been studied~\cite{ahanchi2023massively, andoni2014parallel, azarmehr2025massively, chen2023streaming,czumaj2024fully,epasto2022massively,jayaram2024massively}, they represent only a limited portion of the rich landscape of computational geometry. 
In particular, only a few prior works consider fundamental proximity structures such as well-separated pair decompositions, which are central to many classic geometric algorithms.

While only a few papers explicitly address geometric proximity problems in the MPC model, such problems have been extensively studied in earlier parallel computation models, particularly the PRAM model~\cite{aggarwal1988parallel, amato1994parallel,callahan1993optimal, callahan1995decomposition,cole1992optimal, reif1992optimal}. For additional references, see the survey~\cite{goodrich2017parallel}.
These algorithms typically assume global memory access and fine-grained synchronization, but these assumptions are less realistic in modern distributed systems.
In contrast, the MPC model more accurately captures the constraints of today’s large-scale computing environments by modeling limited local memory and emphasizing communication-efficient computation.
Prior work~\cite{goodrich2011sorting, karloff2010model} has shown that PRAM algorithms can be simulated in the MPC model under reasonable assumptions. These simulation results can be viewed as implicit progress toward solving proximity problems in the MPC model.
However, since a PRAM algorithm with depth~$t$ requires~$O(t)$ rounds to simulate in the  MPC model~\cite{goodrich2011sorting, karloff2010model}, applying this approach to the aforementioned algorithms would still incur~$\Omega(\log n)$ rounds.
This highlights a fundamental challenge in designing native MPC algorithms that solve geometric proximity problems in sub-logarithmic rounds, fully leveraging the distributed nature of the model. 

\medskip 
In this work, we take a systematic step in this direction by developing a fully scalable constant-round MPC algorithm for computing a well-separated pair decomposition 
in a metric space with a constant \emph{doubling dimension} $\ddim$, i.e., every ball of radius $r$ can be covered by $2^\ddim$ balls of radius $r/2$~\cite{gupta2003bounded}.
\ccheck{We say that two sets $A,B$ of points are \emph{$(1/\eps)$-separated} if $\max\{\diam(A), \diam(B)\} \leq  \varepsilon\cdot \dist{}(A,B)$ for $\varepsilon>0$.
A collection $\wspd =\{\{A_1,B_1\}, \ldots, \{A_k,B_k\}\}$ is called a \emph{$(1/\eps)$-well-separated pair decomposition} ($(1/\eps)$-\textsf{WSPD}) of a point set $P$ if (i) $A_i,B_i\subset P$, and $A_i$ and $B_i$ are $(1/\eps)$-separated for every $i$, (ii) $A_i\cap B_i=\emptyset$ for every $i$, and (iii) $\bigcup_{i=1}^kA_i\otimes B_i=P\otimes P$, where $A\otimes B=\{(x,y) \mid x\in A, y \in B, x\neq y\}$.
We define the \emph{size} of the \textsf{WSPD} as the number of pairs it contains.}
A \textsf{WSPD} compactly represents all pairs of point subsets that are sufficiently far apart, which has numerous applications.
Due to its importance, computing a \textsf{WSPD} was among the earliest topics studied in parallel computational geometry under the PRAM model.
For example, Callahan and Kosaraju~\cite{callahan1993optimal} gave an optimal PRAM algorithm that computes an $O((1/\eps)^dn)$-sized $(1/\eps)$-\textsf{WSPD} of $n$ points in $d$-dimensional Euclidean space in $O(\log n)$ depth using $O(n)$ processors.
This yields an $O(\log n)$-round MPC algorithm via standard simulation~\cite{goodrich2011sorting, karloff2010model}, but it remains open whether a $o(\log n)$-round MPC algorithm is possible.
In contrast, in doubling metric spaces, a linear-sized \textsf{WSPD} is also known to exist~\cite{har2006fast}, yet no efficient PRAM algorithm is known for constructing it.
This makes it particularly challenging to design MPC algorithms in doubling metric spaces.

\subparagraph*{Our main results.}
In this paper, we present the first $O(1)$-round, fully scalable MPC algorithm for computing a $(1/\eps)$-\textsf{WSPD} in a doubling metric space as stated below.


\begin{theorem}\label{thm:intro-dd-algorithm}
    Given a set $P$ of $n$ points in a metric space with a constant doubling dimension $\ddim$ for any $\eps>0$, we can compute a $(1/\varepsilon)$-\textsf{WSPD} of $P$ of size $(1/\eps)^{O(\ddim)}n\log^2 n$ in $O(1)$ rounds, using $(1/\eps)^{O(\ddim)} n\log^3 n$ total space and  $O(n^\delta)$ local memory per machine for any $\delta\in (0,1)$. The algorithm succeeds with $1-1/n^{\Omega(1)}$ probability.
\end{theorem}

Since the $d$-dimensional Euclidean space has doubling dimension $O(d)$, we can construct a WSPD by applying Theorem~\ref{thm:intro-dd-algorithm}.
However, by exploiting several classical geometric tools specific to Euclidean spaces, we can compute a linear-sized WSPD deterministically.

\begin{theorem}\label{thm:intro-eclidean-algorithm}
    Given a point set $P$ of size $n$ in $\mathbb R^d$ for any constant $d\geq 1$ and for any $\varepsilon>0$, one can construct a $(1/\eps)$-\textsf{WSPD} of $P$ of size $(1/\varepsilon)^{O(d)} n$ in $O(1)$ rounds, using $(1/\varepsilon)^{O(d)} n$ total space and $O(n^\delta)$ local space\footnote{\ccheck{More precisely, the total space and local space are $(1/\varepsilon)^{O(d)} \cdot (dn)$ and $O((dn)^\delta)$, respectively, which ensures full scalability. Since $d$ is treated as a constant, we omit the dependence on $d$.}} per machine for any $\delta \in (0,1)$.
\end{theorem}

Our result improves upon the simulated PRAM algorithm of~\cite{callahan1993optimal} in two key aspects: it reduces the round complexity from $O(\log n)$ to $O(1)$, and extends beyond constant-dimensional Euclidean spaces to general metric spaces with constant doubling dimension.

\subparagraph*{Applications.} 
As applications, we present fully scalable MPC algorithms for constructing a spanner, computing the closest pair, approximating the diameter, and constructing all $k$-nearest neighbors as stated below.
Let $P$ be a set of $n$ points in a metric space $\mathcal X$, and for any $\eps>0$,
let $f_{\mathcal X}(n,\eps)$ be the size of a $(1/\eps)$-\textsf{WSPD} of $P$ computed by our MPC algorithms.

\begin{itemize} 
    \item \textbf{Closest pair:} One can construct the \emph{exact} closest pair of $P$ in $O(1)$ rounds, using $f_{\mathcal X}(n,\eps)$ total space and $O(n^{\delta})$  local space per machine for any $\delta \in (0,1)$.

    \item \textbf{Diameter:}
    One can construct a $(1-\varepsilon)$-approximate diameter of $P$ in $O(1)$ rounds, using $f_{\mathcal X}(n,\eps)$ total space and $O(n^{\delta})$  local space per machine for any $\delta \in (0,1)$.

        \item \textbf{Spanner:}
One can construct a $(1+\varepsilon)$-spanner of $P$ with $f_{\mathcal X}(n,\eps)$ edges in $O(1)$ rounds, using $f_{\mathcal X}(n,\eps)$ total space and $O(n^{\delta})$  local space per machine for any $\delta \in (0,1)$.
    
    \item \textbf{All $k$-nearest neighbors:}
    One can construct $k$-nearest neighbors of each point of $P$ in $O(k)$ rounds in total, using $k\cdot f_{\mathcal X}(n,\eps)$ total space and $O(n^{\delta})$  local space per machine for any $\delta \in (0,1)$, if $\mathcal X$ is a Euclidean space.    
\end{itemize}

\subparagraph*{Comparison with prior work.}
As these problems are fundamental in computational geometry, several results are known. Although no prior work explicitly studies the closest-pair or diameter problem in the MPC model, they can be handled by classical geometric techniques. We also note that the closest-pair problem reduces to all-nearest neighbors. 
Overall, except for a few Euclidean-specific cases, our results match or outperform the best known MPC algorithms. Moreover, our approach is \emph{unified}: the same WSPD-based framework simultaneously yields algorithms for all the problems discussed below.

\begin{itemize} 
    \item \textbf{Closest pair:} To the best of our knowledge, no prior MPC algorithm is known for the closest pair problem in doubling metric spaces.
    In contrast, it can be solved in $\mathbb{R}^d$ with $d \le 2$ in $O(1)$ rounds 
    by constructing the Delaunay triangulation~\cite{goodrich1997randomized, nath2016massively}.
    
    \item \textbf{Diameter:} A $2$-approximation of the diameter can be found in $O(1)$ rounds by selecting an arbitrary point and computing its farthest neighbor. In $\mathbb{R}^d$ for fixed $d$, a $(1-\varepsilon)$-approximation follows from evaluating the maximum directional width over a $(1/\varepsilon)^{O(d)}$ net on the unit sphere~\cite{agarwal1992farthest}, which can be implemented in $O(1)$ rounds in the MPC model.
    
    \item \textbf{Spanner:} To the best of our knowledge, no prior MPC algorithm is known for constructing spanners in doubling metric spaces.
    In contrast, for Euclidean spaces, $t$-spanners with $O(n^{1+1/t^2})$ edges can be constructed in $O(1)$ rounds~\cite{cohen2022massively, epasto2022massively}, as explicitly stated in~\cite{jayaram2024massively}.
    In high-dimensional Euclidean spaces, they outperform ours in terms of round complexity, whereas our spanner has fewer edges in low-dimensional Euclidean spaces.
    
    \item \textbf{All nearest neighbors:} To the best of our knowledge, no prior MPC algorithm is known for solving the all-nearest neighbors problem in doubling metric spaces.
    In contrast, for Euclidean spaces, all $c$-\emph{approximate} nearest neighbors can be computed in $O(1)$ rounds with total space complexity $n^{1/c} \cdot \tilde O(n)$~\cite{czumaj2024fully}.
    In high-dimensional settings, this result outperforms ours in terms of total space complexity, whereas our algorithm computes \emph{exact} solutions with lower space usage in low-dimensional Euclidean spaces.
    In such low-dimensional settings, the previously best-known algorithm for approximate nearest neighbors achieves $O(1)$ rounds with $O(n)$ total space~\cite{agarwal2016parallel}.
    Our algorithm improves upon this by computing exact nearest neighbors within the same round and space complexity.
    
    \item \textbf{All $k$-nearest neighbors.}
    The all $k$-nearest neighbors problem can be solved in the binary-fork model~\cite{blelloch2022parallel}, which can be simulated in the MPC model. Its performance depends on the spread and the expansion of $P$, whereas our algorithm works for arbitrary point sets.
    If we simulate their algorithm in the MPC model on a point set with bounded spread and bounded expansion, the number of rounds is $O(n^\gamma)$ for a constant $\gamma <1$.
    In contrast, our algorithm uses $O(k)$ rounds for any point set.
\end{itemize}

\subparagraph*{Additional contribution: New hierarchical decomposition.}
To construct a \textsf{WSPD} efficiently, we build a hierarchical decomposition of $P$, which is also our contribution that may be of independent interest.
In doubling metric spaces, we introduce a new hierarchical structure called the \emph{partition cover tree}, which can be built in $O(1)$ rounds.
A central tool for handling doubling metric spaces is a \emph{net tree}~\cite{har2006fast}. 
While its uncompressed variant can be built in $O(1)$ rounds~\cite{andoni2014parallel}, it remains open whether the \emph{compressed} net tree, which is essential for handling point sets with unbounded spread, can be built in $O(1)$ rounds. Here, the spread of a point set is the ratio between its largest and smallest pairwise distances.
\emph{Our work bypasses this bottleneck by introducing the partition cover tree.}
Although it does not satisfy the packing property or the invariant that radii halve at each level unlike the net-tree, it retains a key geometric feature of the compressed net-tree: any ball can be covered by a small number of cells with comparable diameter, and the size is linear only in the input size.
We believe this is of independent interest and may find broader applications in proximity problems over doubling metrics.
For example, this may serve as a foundation for data structures that approximate local regions using a small number of representations such as consistent hashing and locality-sensitive hashing~\cite{bhaskara2018distributed}.

In Euclidean spaces, we present an MPC algorithm for constructing a \emph{compressed quadtree} of linear size in $O(1)$ rounds.
Tree-based hierarchical decompositions such as split trees, $kd$-trees, BBD-trees, randomly shifted quadtrees, partition tree for range searching have been designed in the MPC or PRAM model~\cite{agarwal2016parallel,andoni2014parallel,callahan1993optimal,callahan1995decomposition}, and have been used in a wide range of parallel geometric algorithms. In a similar sprit, our compressed quadtree might serve as a foundational structure for other MPC algorithms for proximity problems. 

\section{Preliminaries} \label{sec:pre}
Consider a metric space $\metric$.
For any two points $p$ and $q$, we denote their distance in $\metric$ by $\dist{\metric}(p,q)$.
For convenience, we use $\dist{}(p,q)$ for $\dist{\metric}(p,q)$ if the underlying metric space is clear from the context.
We assume that one can compute $\dist{\metric}(p,q)$ in $O(1)$ time for any metric space $\metric$.
Let $B_{\dist{}}(x,r) = \{y \in \metric \mid \dist{\metric}(x,y) \leq r\}$.
For convenience, we use $B(x,r)$ to denote $B_{\dist{}}(x,r)$ if the underlying distance function is clear from the context.
And, for a ball $B$, we use $c(B)$ and $r(B)$ to denote the center and radius of $B$, respectively.
Let $P$ be a finite point set in the metric space $\metric$.
For any two subsets $X, Y\subseteq P$, we define their distance as the minimum distance between any pair of points $x \in X$ and $y \in Y$.
That is, $\dist{}(X,Y)=\min_{x \in X, y\in Y}\dist{}(x,y)$.
When $X=\{x\}$ is a singleton, we alternatively use $\dist{}(x,Y)$ to denote $\dist{}(X,Y)$.
We denote the \emph{diameter} of $P$, defined as the maximum pairwise distance between the points in $P$ by $\diam(P)=\max_{p,q\in P}\dist{}(p,q)$. 

In this paper, we consider two fundamental metric spaces: doubling metric spaces and Euclidean space.
The doubling dimension $\ddim(\metric)$ of a metric space $\metric$ is the smallest value $d$ such that 
every ball $B$ can be covered by $2^d$ balls of half the radius of $B$.
A useful consequence of the doubling property is the packing property: any set of points in $B(x, r)$ that are pairwise at distance greater than $r'$ has size at most $(r/r')^{O(\ddim)}$.

Throughout this paper, we consider several tree structures.  
Let $u$ be a node in a tree, and $\lvl{u}$ be the level of $u$, which is the distance from $u$ to the root of the tree.
We denote the parent of $u$ by $\parent{u}$.
We use $[m]$ to denote $\{1,2,\ldots,m\}$ for an integer $m\geq 1$.

\subsection{Well-Separated Pair Decomposition.}\label{sec:pre-wspd}

In a metric space with doubling dimension $\ddim$, any point set of size $n$ has a $(1/\eps)$-well-separated pair decomposition consisting of $(1/\eps)^{O(\ddim)}n$ pairs~\cite{callahan1995decomposition, har2006fast}.
While a doubling metric space has a \textsf{WSPD} of linear size, the sum of the sizes of all point pairs in the \textsf{WSPD} can be quadratic. 
To avoid explicitly listing all points in each pair, we use a compact representation: we represent each set as an arbitrary representative point in the set and the maximum distance from its representative point to any point in the set.\footnote{In the Euclidean case, we construct a \textsf{WSPD} using a compressed quadtree.
Since each set corresponds to a cell in the quadtree, we can simply use the center and diameter of that cell to represent this set.}

\begin{algorithm}[tb]
    \centering 
    \caption{$\textsf{genWSPD}(u,v)$}
	\begin{algorithmic}[1] 
        \State Assume $r_u\leq r_v$.
         Otherwise, exchange the roles of $u$ and $v$.
        \If{$\max\{r_u, r_v\} \leq (\varepsilon /8)\cdot \dist{}(p_u,p_v)$} 
            \State \textbf{Return} $\{u,v\}$
        \Else
        \State Denote the children of $v$ by $v_1,\ldots v_i$. 
        \State \textbf{Return} $\bigcup_{j=1}^i\textsf{genWSPD}(u, v_j)$
        \EndIf
	\end{algorithmic}
	\label{alg:har-wspd}
\end{algorithm}

\subparagraph*{Construction.}
We can compute a $(1/\varepsilon)$-\textsf{WSPD} efficiently once we have a \emph{partition tree}. In the Euclidean case, one can use a quadtree as the partition tree~\cite{callahan1995dealing,callahan1995decomposition,har2011geometric},
while in a doubling metric space, a net tree can serve as the partition tree~\cite{har2006fast}.
We define a \emph{partition tree} as a tree satisfying the following properties: 
(1) each node $v$ corresponds to a subset $P_v$ (called a \emph{piece}) of $P$, 
(2) the pieces of the children of a node $v$ form a partition of $P_v$, 
(3) the piece of the root is $P$,
and (4) the piece of a leaf consists of a single point. 
Each vertex $v$ is associated with a representative $p_v$ and the maximum distance $r_v$ from $p_v$ to any point in $P_v$.

Algorithm~\ref{alg:har-wspd} describes how to compute a \textsf{WSPD} using a partition tree $T$. In particular, we run Algorithm~\ref{alg:har-wspd} with a pair consisting of two copies of the root of $T$. 
This algorithm always returns a $(1/\eps)$-\textsf{WSPD}, but its size depends on the properties of the underlying partition tree. 
More specifically, the algorithm constructs a $(1/\eps)$-\textsf{WSPD} by recursively checking whether 
two nodes satisfy the following, starting from 
the root node paired with itself.
\begin{equation}\label{equ:cell-separated}
    \max\{r_u, r_v\} \leq (\varepsilon/8) \cdot \dist{}(p_u,p_v).
\end{equation}
If this holds, returns the pair $\{u,v\}$ as a pair of the \textsf{WSPD};
otherwise, assuming without loss of generality that $r_u\leq r_v$, it recursively refines the pair by replacing $v$ with its children: for each child $v'$ of $v$, the algorithm recurses on the pair $\{u, v'\}$.

\subsection{Massively Parallel Computation (MPC) Model}\label{sec:pre-MPC}
In the MPC model, an input of size $n$ is arbitrarily distributed over $\numM$ machines, each with a local memory of $O(\localSize)$ words.
The number of machines is typically $m=O(n/\localSize)$ or slightly larger.
We say that an MPC algorithm is \emph{fully scalable} if it works for $\localSize = n^{\delta}$ with an arbitrary constant $\delta \in (0,1)$.
We assume that a word can store a real number with arbitrary precision as in~\cite{agarwal2016parallel}, which is a standard assumption for geometric problems.
We assume that each machine has a unique ID from $[m]$.
The computation of an MPC algorithm proceeds in \emph{synchronous rounds}, each consisting of a \emph{local computation phase} followed by a \emph{communication phase}.
In the computation phase, each machine performs local computations on its local data.
In the communication phase, each machine sends messages to other machines.
Each machine is allowed to send and receive a total of $O(\localSize)$ words per round.

Throughout this paper, we use the following subroutines that can be implemented in $O(1)$ MPC rounds: sorting, predecessor, indexing, minimum/maximum, broadcast.
For an integer $a$, let $M_a$ be the machine with ID $a$. The following subroutines can be implemented in $O(1)$ rounds in the MPC model.

\begin{enumerate} 
    \item \textbf{\textsf{Sorting}~\cite{goodrich1999communication, goodrich2011sorting}:} Given $n$ comparable items distributed arbitrarily across the machines, one can sort and index all items in $O(1)$ rounds.
    \item \textbf{\textsf{Predecessor}~\cite{andoni2018parallel, goodrich2011sorting}:} Given pairs $(x_1,y_1),(x_2,y_2), \ldots, (x_n,y_n)$, distributed arbitrarily across the machines, where every item $x_i$ has a total ordering and $y_i\in\{0,1\}$, every item $x_j$ can be stored together with the item $x_{i_j}$ and its index $i_j$, where $x_{i_j}$ is the largest item such that $x_{i_j}\leq x_j$ and $y_{i_j}=1$.
    \item \textbf{\textsf{Indexing}~\cite{andoni2018parallel, goodrich2011sorting}:} Suppose that the sets $S_1,S_2,\ldots,S_k$ of total $n$ items are stored across the machines.
    In $O(1)$ rounds, every item $x$ can be stored together with its rank within the set storing $x$ in the machine that originally has $x$.
    \item \textbf{\textsf{Minimum/Maximum}~\cite{dinitz2019massively}:} Given two IDs $a,b$ of the machines for $a\leq b$, the minimum or maximum of items stored in $M_a,M_{a+1},\ldots M_b$ can be computed in $O(1)$ rounds. 
    \item  \textbf{\textsf{Broadcast}~\cite{dinitz2019massively}:}
    Given two IDs $a,b$ of the machines for $a\leq b$, a machine can send a message of size $O(\localSize^{\alpha})$ to $M_a,M_{a+1},\ldots M_b$ in $O(1)$ rounds, where $\alpha<1$ is a constant.
\end{enumerate}

\subsection{Geometric Proximity Problems}
We define the four fundamental problems that we address as applications of the \textsf{WSPD}.
Let $P$ be a point set in a metric space $\mathcal X$ with distance function $\dist{}$.

\begin{itemize} 
\item \textbf{Spanner:} A $t$-\emph{spanner} of $P$ is
a compact representation of the distances between all the pairs of points of $P$.
It is defined as a graph $G=(P,E)$ such that $\dist{}(q,s)\leq d_G(q,s)\leq t\cdot \dist{}(q,s)$,
    where $d_G(q,s)$ is the distance between $q$ and $s$ in $G$. 
    Thus our goal in this problem is to compute a spanner with a small number of edges. 

\item{\textbf{Diameter and closest pair:}}
The \emph{diameter} and \emph{closest pair} are fundamental geometric characteristics of a point set.
Recall that the diameter of a point set $P$ is defined as the maximum pairwise distance among the points, i.e., $\max_{p, q \in P} \dist{}(p, q)$.
In this work, we aim to compute a $(1-\varepsilon)$-approximate diameter, that is, the distance between some pair of points in $P$ that is at least $(1-\varepsilon) \cdot \diam(P)$.
Conversely, the closest pair is the pair of distinct points in $P$ with the minimum distance between them, i.e., a pair $(x, y)$ such that $\dist{}(x, y) = \min_{p \ne q \in P} \dist{}(p, q)$.

\item{\textbf{$k$-nearest neighbors:}}
For a point $p \in P$, the \emph{nearest neighbor} in $P$ is the point in $P \setminus \{p\}$ that is closest to $p$.
The \emph{$k$-nearest neighbors} ($k$-NN) problem asks to compute, for each point $p \in P$, the $k$ closest points in $P \setminus \{p\}$ to $p$. 
They are fundamental tasks in computational geometry and have numerous applications in clustering, classification, and data analysis.
\end{itemize}

\subparagraph{Overview.}
In both doubling metric spaces and Euclidean spaces, our algorithm follows a unified two-step framework.
First, we construct a partition tree $T$ over the input point set $P$. Then, we generate pairs of nodes $(u, v)$ in $T$ such that the corresponding piece pair $(P_u, P_v)$ form a  \textsf{WSPD}.
For doubling metric spaces, we introduce a new structure called the \emph{partition cover tree}, specifically designed for our setting.
In contrast, for Euclidean spaces, we use a compressed quadtree as the partition tree.


\section{Technical Overview}
\subsection{Overview of WSPD Construction in Doubling Metric Spaces}\label{sec:overview-doubling}
Let $P$ be a set of $n$ points in a metric space with a constant doubling dimension $\ddim$.
We introduce our approach with three ingredients: \emph{partition cover trees}, \emph{radius perturbation}, and \emph{$\alpha$-approximation}. 
We compute a \textsf{WSPD} in two steps: first, by constructing a \emph{partition cover tree} and then by constructing a \textsf{WSPD} using this tree.
Given the partition cover tree, we can construct a \textsf{WSPD} by parallelizing Algorithm~\ref{alg:har-wspd}.
Details can be found in Section~\ref{sec:doubling}.

\subparagraph*{First ingredient: Partition cover tree.}
Given a partition tree $T$ of $P$, we can compute a \textsf{WSPD} by traversing the tree as in sequential algorithm by Har-Peled and Mendel~\cite{har2006fast}. 
However, the size of the constructed \textsf{WSPD} depends on the properties of the underlying partition tree.
For instance, Har-Peled and Mendel~\cite{har2006fast} used a net tree 
to ensure that the number of pairs in the \textsf{WSPD} is  $O(n)$.
However, it is unclear whether their algorithm can be adapted to the MPC model, as it relies on constructing a \emph{greedy permutation} of $P$.\footnote{A greedy permutation is a sequence of points $\langle p_1,\ldots,p_n\rangle$ from $P$, associated with a sequence $\langle r_1,\ldots,r_n\rangle$ of radii, such that
$P\subseteq \cup_{\ell=1}^k B(p_\ell, r_k)$, and $\min_{1\leq i\leq j\leq k}\dist{}(p_i,p_j)=r_{k-1}$ for any integer $k$.
Here, the choice of each point $p_i$ depends on the prefix $\langle p_1,\ldots,p_{i-1}\rangle$ of the sequence.}

Instead, we present a new type of partition tree that enables the construction of a \textsf{WSPD} of size $(1/\eps)^{O(\ddim)} n\log^2 n$.
We show that any randomized partition tree satisfying the following property yields a \textsf{WSPD} of size $(1/\eps)^{O(\ddim)}n\log^2 n$. We call a partition tree a \emph{partition cover tree} if it satisfies the following small cover property. For its proof, see Lemma~\ref{lem:size-candidate} and Lemma~\ref{lem:size WSPD} in Section~\ref{sec:doubling}.

\begin{property}[Small Cover Property]
For a ball $B$, a set $\{v_1,v_2,\ldots,v_k\}$  of nodes of a randomized partition tree $T$ is called a \emph{cover} of $B$ induced by $T$ if
the union of $P_{v_1},\ldots,P_{v_k}$ contains $B$,
and 
$\diam(P_{v_i})\leq \diam(B)/10$ for all nodes $v_i$.
We say $T$ has the \emph{small cover property} if 
for every ball $B$ in the metric space $\metric$,
the expected size of the minimum-cardinality cover of $B$ is $O(\log n)$.
\end{property}

As mentioned earlier, a major obstacle in designing MPC algorithms for doubling metric spaces with unbounded spread has been the absence of an efficient MPC construction of the \emph{compressed net tree}.
Although the partition cover tree does not satisfy the packing property or the invariant that radii halve at each level, it shares the small cover property.

\subparagraph*{Second ingredient: Radius perturbation.} 
We can construct a partition tree by partitioning 
 $P$ into subsets in a top-down manner. Here, we describe the construction in the sequential model, and later we briefly show how to implement it in the MPC model using $\alpha$-approximation. 
 Imagine that we compute a ball $B$ containing $|P|/2$ points of $P$.
Then we recursively construct the tree of $P\cap B$
and the tree of $P\setminus B$, and merge them into a single tree by creating the root corresponding to $P$. 
Clearly, the constructed tree is a partition tree of height $O(\log n)$, but it does not necessarily satisfy the small cover property.
To see this, observe that a small ball can be split into two subsets by $B$. 
Moreover, these subsets can be further subdivided at higher levels. 
As a result, a small ball may be split into $\Theta(n)$ pieces before reaching levels where the diameters of the pieces become smaller than the diameter of the ball itself.


To avoid the case that a small ball is separated into many pieces, we choose the radius of a ball in the construction of the partition tree randomly. 
Instead of choosing a ball containing $|P|/2$ points of $P$,
we first compute a smallest-radius ball $B$ containing at least $|P|/4^{O(\ddim)}$ points from $P$. 
Then $|P|/4^{O(\ddim)} \leq |B(c, 2r) \cap P| \leq |P|/2^{O(\ddim)}$ by the choice of $B$,
where $c$ and $r$ denote the center and radius of $B$, respectively.
The second inequality holds, since $B(c,2r)$ can covered by at most $2^{O(\ddim)}$ balls of radius slightly smaller than $r$.
Thus, to ensure that a ball separates $P$ in a balanced way, we may choose any radius from $[r,2r]$.  We choose a value $\theta$ from $[0,1]$ uniformly at random, and partition $P$ into two subsets 
using $B^*=B(c, (1+\theta)\cdot r)$.
Then the probability that a ball of radius $r'$ is separated by $B^*$ is at most $r'/r$, which is very small if $r'\ll r$.
Then we recurse on $P\cap B^*$ and $P\setminus B^*$. 
For a technical reason, if $r\geq \bar r/10$, where 
$\bar r$ is the radius of the separating ball of its parent, we further partition $P$ using an arbitrary ball of radius $\bar r$ 
with radius perturbation, and then 
construct a binary partition tree of this partition first, and then recurse on each piece of the partition.
Then we can show that the resulting tree $T$ satisfies the small cover property.
For its proof, see Lemma~\ref{lem:total-sum} and Lemma~\ref{lem:expected size} in Section~\ref{sec:doubling}.

\medskip 
In the MPC framework based on $\alpha$-approximation, which will be described later, 
it is crucial not only to obtain a partition satisfying the small cover property,
but to obtain it as a \emph{hierarchical} tree of \emph{height} $O(\log n)$. 
This $O(\log n)$-height hierarchical structure is what enables us to construct the tree in $O(1)$ MPC rounds in a top-down manner. 
Without such constraints, one may satisfy the covering property in other ways like~\cite{DBLP:conf/icalp/Filtser20}, but these constructions do not appear to admit $O(1)$-round parallelization in the MPC model.

Ensuring both hierarchy and bounded height, however, makes the analysis considerably more delicate.  
Although radius perturbation has appeared in prior geometric decompositions,
those arguments rely on the property that the diameters of the pieces decrease by a constant factor along a root-to-leaf path. However,
in our construction, the radii of the separating balls do not necessarily decrease by a constant factor along root-to-leaf paths,
and therefore the standard arguments used in classical radius-perturbation schemes do not carry over.
This is a main technical challenge we overcame in this paper.

\subparagraph*{Third ingredient: $\alpha$-approximation.}
To construct the partition cover tree $T$ in $O(1)$ rounds, we compute the first $(\log s)/2$ levels of $T$ at once using an $(1/s^{1/4})$-approximation of a proper range space, which is our third ingredient. A similar approach is also used in~\cite{agarwal2016parallel} to construct a range tree in the Euclidean space in the MPC model.
An $\alpha$-approximation is originally introduced for approximate counting and geometric estimation.
Here, we review several notions for defining the $\alpha$-approximation. See also~\cite{har2011geometric} for further details.
A \emph{range space} $S$ is a pair $(X, \mathcal{R})$, where $X$ is a ground set, and $\mathcal{R}$ is a family of subsets of $X$, called \emph{ranges}.

\begin{definition}[$\alpha$-approximation]\label{def:a-approx}
    Given a range space $(X, \mathcal{R})$ and a parameter $0 \leq \alpha \leq 1$, a subset $Y \subseteq X$ is an \emph{$\alpha$-approximation} of $(X, \mathcal{R})$ if, for every $R \in \mathcal{R}$,
    \[
    \left| \frac{|Y \cap R|}{|Y|} - \frac{|R|}{|X|} \right| \leq \alpha.
    \]
\end{definition}

An $\alpha$-approximation of a range space $S$ can be obtained via random sampling, where the required sample size depends on the VC dimension of $S$.
Or, alternatively, if the number of ranges is small, we can use the following lemma. Its proof can be found in the appendices.


\begin{lemma}[Folklore]\label{thm:a-approx-sampling-poly}
    There exists a constant $c > 0$ such that a set obtained by $m$ random independent draws from $X$ 
    forms an $\alpha$-approximation of $(X, \mathcal{R})$ with probability at least $1 - \psi$, regardless of VC dimension, where
    \[
    m=\frac{c}{\alpha^2} \left( \log|\mathcal R| + \log \frac{1}{\psi} \right).
    \]
\end{lemma}
\begin{proof}
    Let $R$ be a fixed range of $\mathcal R$.
    Let $p_R=|R|/|X|$ be the true fraction of $X$ in $R$.
    Let $Z_1,Z_2,\ldots,Z_m$ be the indicator random variables where $Z_i=1$ if the $i$th sampled point is contained in $R$, and $Z_i=0$, otherwise.
    Then each $Z_i$ is a Bernoulli random variable with mean $\mathbb{E}[Z_i]=p_R$.
    Let $\hat p_R=(1/m)\sum_{1}^m Z_i$.
    By the additive Chernoff bound (or the Hoeffding's inequality), for any $\alpha>0$,
    we have 
    \[
    \prob[|\hat p_R-p_R|>\alpha]\leq 2{e^{-2\alpha^2m}}. 
    \]
    That is, the success probability for $R$ is high. 
    Now we take a union bound over all ranges $R$ on $\mathcal R$.
    \[
    \prob[ \exists R\in\mathcal R \text{ s.t } |\hat p_R-p_R|>\alpha ]\leq |\mathcal R|\cdot 2{e^{-2\alpha^2m}}. 
    \]
    By the choice of $m$, we can show that the success probability is at least $1-\psi$.
\end{proof}

We use the $\alpha$-approximation to parallelize the construction of a partition cover tree.
Consider the range space $S=(P,\mathcal R)$ where $\mathcal R = \{ \cap_{B\in \mathcal B}B \setminus \cup_{B'\in\mathcal B'}B'\}$, and $\mathcal B$ and $\mathcal B'$ are sets of $O(\log n)$ balls. Since the size of $\mathcal R$ is $n^{O(\log n)}$, a sample $Y$ of size $O(\sqrt s\log^2 n)\leq O(s)$ is an $(s^{1/4})$-approximation of $S$. 
Then we construct a partition tree of $Y$ until every leaf corresponds to a set of $O(1)$ points of $Y$. This also partitions $P$ into $O(s)$ subsets, each is contained in $\mathcal R$, accordingly. The fact that $Y$ is an 
$(s^{1/4})$-approximation of $S$ implies that each subset in the partition contains $O(n/s^{1/4})$ points of $P$.
Therefore, by repeating this $O(1)$ times, we can obtain a partition cover tree in $O(1)$ rounds.

\begin{theorem}\label{thm:partition-cover-tree-construction}
    Given a set $P$ of $n$ points in a doubling metric space with bounded doubling dimension $\ddim$, one can construct the partition cover tree of $P$ of $O(n)$ size and complexity.
    It takes $O(1)$ rounds with high probability, and uses $O(n)$ total space and $O(n^\delta)$ local space per machine for any $\delta \in (0,1)$.
\end{theorem}

\noindent
\textit{Remark.}
Both the partition cover tree and the tree of Andoni et al.~\cite{andoni2014parallel} can be viewed as variants of a net tree in the sense that they preserve several key structural properties of net trees.
However, neither construction is a full-fledged net tree in the classical sense.
A main advantage of our structure is that we can construct it in $O(1)$ MPC rounds even when \emph{the spread is not bounded.} 
In the following, we discuss the precise relationship between the two structures, and which properties they share or differ in.
Andoni et al.~\cite{andoni2014parallel} also used the radius perturbation scheme to construct a tree with key structural properties of an \emph{uncompressed} net-tree in the MPC model.
Their construction assumes that the input point set has bounded spread.
We highlight two main differences between their tree and our partition cover tree below.

\begin{itemize}
    \item First, the desired properties of the resulting tree differ.
    The authors in~\cite{andoni2014parallel} designed their tree as part of the MST construction.
    However, their analysis does not immediately guarantee that the tree satisfies the small cover property.
    In contrast, our construction explicitly aims to achieve the small cover property, which requires new technical insights.
    This is particularly challenging in our setting because the radii of separating balls do not necessarily decrease by half along root-to-leaf paths.
    Such non-monotonicity arises when a node lies outside the separating ball of its ancestor.
    Establishing that our tree satisfies the small cover property is therefore a key technical contribution of this work.
    \item Second, the construction itself differs from theirs.
    The authors in~\cite{andoni2014parallel} recursively subdivide each piece into $2^{O(\ddim)}$ smaller pieces of approximately half the diameter.
    Given the assumption that the spread is polynomially bounded in $n$, the resulting tree has height $O(\log n)$, allowing its construction in $O(1)$ MPC rounds.
    However, in the general case without the bounded spread assumption, this method yields a tree of height $O(n)$, making the same constant-round parallelization infeasible.
    In contrast, our approach selects a separating ball $B$ that balances the size of each piece. 
    This strategy ensures that the resulting tree has height $O(\log n)$ regardless of the spread, enabling its construction in $O(1)$ MPC rounds via an $\alpha$-approximation, as described later.
\end{itemize}

\subparagraph*{Constructing a \textsf{WSPD}.}
Given a partition cover tree $T$, we construct a \textsf{WSPD} of size $(1/\eps)^{O(\ddim)}n\log^2 n$ with at least constant probability.
We use the observation from~\cite{har2011geometric} for Euclidean spaces with a compressed quadtree as the partition tree:
for any pair $(u,v)$ in the \textsf{WSPD} constructed by Algorithm~\ref{alg:har-wspd}, such that the recursion on $\{u,v\}$ is invoked from the recursion on $\{u,\parent{v}\}$, the following conditions hold.
$\dist{}(A_{\parent{u}}, A_{\parent{v}})\leq (1/\eps)\cdot \diam(A_{\parent{v}})$, and $\diam(A_{\parent{v}})\leq \diam(A_{\parent{u}})$, where $A_u$ is the piece corresponding to $u$.
We show that analogous conditions hold for doubling metric spaces when the partition tree is the partition cover tree: for any pair $\{u,v\}$ in the \textsf{WSPD} constructed by Algorithm~\ref{alg:har-wspd}, $\dist{}(p_v,p_{\parent{u}})\leq (8/\eps)r_{\parent{u}}$, and $r_{\parent{u}}\leq r_{\parent{v}}$.
For its proof, see Lemma~\ref{lem:returned pair is a candidate pair} in Section~\ref{sec:doubling}.
We refer to a node $u\in T$ as a \emph{pairing candidate} of a node $v\in T$ if $v$ satisfies these conditions.

To construct a \textsf{WSPD}, we identify the \emph{candidate pairs} consisting of a node in $T$ and one of its pairing candidates.
We perform this efficiently in parallel by grouping $O(\log \localSize^{1/4})$ consecutive levels of $T$.
Within this range of levels, the nodes induce several disjoint subtrees of $T$.
For each pair of subtrees, we identify all candidate pairs formed by one node from each subtree.
This computation fits in a single machine, since every subtree in the range has size at most $O(\localSize^{1/4})$.
However, the total number of candidate pairs from a pair of subtrees may exceed local memory.
To avoid this, we utilize the pre-assignment strategy: we first count the number of such pairs and allocate a sufficient number of machines to construct them.
Furthermore, we can amplify the success probability to $1 - 1/n^{\Omega(1)}$ by increasing the total space by a factor of $O(\log n)$, without increasing the local memory.
It concludes Theorem~\ref{thm:intro-dd-algorithm}.

\subsection{Overview of WSPD Construction in Euclidean Space}\label{sec:overview-Euclidean}
Now we consider a set $P$ of $n$ points in $\mathbb R^d$.
As in the doubling metric case, we compute a \textsf{WSPD} in two steps: first, by constructing a partition tree, and then by constructing a \textsf{WSPD} using this tree.
However, instead of the partition cover tree, we construct a compressed quadtree, which leads to an improvement in the size of the \textsf{WSPD}.
Given the compressed quadtree, we can construct a \textsf{WSPD} similarly to the doubling metric case, while the process is considerably simpler in the Euclidean case.
We can assume that $P \subseteq [0,1)^d$, since all coordinates can be uniformly shifted and scaled to fit within $[0,1)$, which can be done in $O(d)=O(1)$ rounds by computing and broadcasting the maximum value of each coordinate.

As in~\cite{har2011geometric}, we use a compressed quadtree to construct a \textsf{WSPD}.
Consider the axis-aligned unit hypercube $\cell$ containing $P$.
The unit hypercube can be hierarchically partitioned into smaller axis-aligned hypercubes, where each cell is recursively subdivided into $2^d$ cells of half the side length until every cell contains at most one point in $P$.
The \emph{(regular) quadtree} of $P$ is a tree $\qt$ whose nodes correspond to a cell of the hierarchical partition of $\cell$, and whose non-leaves have $2^d$ children, satisfying the following properties.
Let $\nodeCell{v}$ denote the corresponding cell of a node $v$ of $\qt$.
The cells corresponding to the children of a non-leaf $u$ form a subdivision of $\nodeCell{u}$ into cells of half the side length.
Each leaf of $\qt$ corresponds to a cell containing at most one point in $P$. 
Throughout this paper, if $\nodeCell{v}$ of a leaf $v$ contains a point $p\in P$, we consider $\nodeCell{v}$ as $p$ itself.

Note that the height of the quadtree can be unbounded in the worst case if the spread of $P$ is  not bounded. 
Thus, instead of a quadtree, we need to use a \emph{compressed quadtree}, which is a compact representation of the quadtree. 
Intuitively, the compressed quadtree of $P$ is the tree whose nodes are the nodes in a quadtree of $P$ whose corresponding cells contain at least one point in $P$, and no child is corresponding to a cell containing the same subset of $P$.
More specifically, the compressed quadtree of $P$ is obtained from the (regular) quadtree by removing all nodes whose cells are empty (i.e., contain no point of $P$) and contracting any node whose cell contains the same point set as its only child.
Observe that each leaf of the compressed quadtree corresponds to a unique point in $P$.

\begin{figure}
    \centering
    \includegraphics[width=0.75\textwidth]{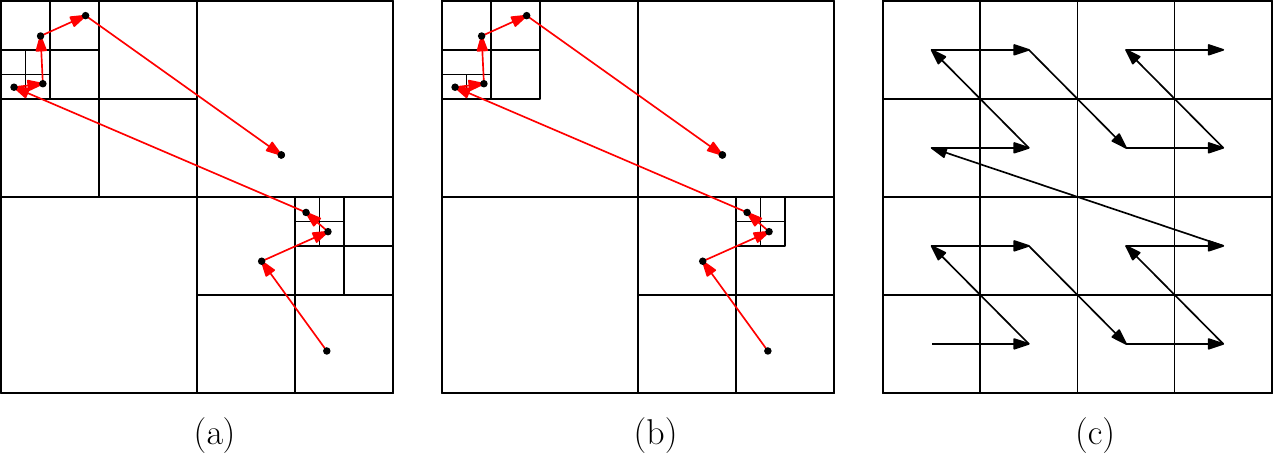}
        \caption{\small (a) Illustration of the quadtree. (b) Illustration of the compressed quadtree.
        In (a) and (b), the ordering of the points colored red is the $\mathcal{Z}$-order.
        (c) Illustration of the $\mathcal{Q}$-order among the cells of the same side length.
    \label{fig:qt}}
\end{figure}

\subparagraph*{Key idea: Using $\mathcal Q$-order and $\mathcal Z$-order.}
To construct the compressed quadtree, we make use of a total ordering of cells of $T$ and points of $P$; see~\cite{har2011geometric}.
Consider a quadtree $\qt$ of $P$, and its depth-first search (DFS) traversal, where the children of each node are always visited in the same pre-defined order.
The \emph{$\mathcal{Q}$-order} of $\qt$ is a total ordering of the cells corresponding to nodes in $\qt$ according to the DFS order of $\qt$.
See Figure~\ref{fig:qt}(c).
Notice that the $\mathcal{Q}$-order on the compressed quadtree is inherited directly from its underlying (regular) quadtree.
Moreover, the total ordering of $P$ obtained by restricting the $\mathcal{Q}$-order to the points of $P$ is known as the \emph{$\mathcal{Z}$-order} of $P$.
See Figure~\ref{fig:qt}(a-b).

As in~\cite{har2011geometric}, we assume that both the floor operation $\lfloor \cdot \rfloor$ and the bit-index operation $\text{bit}_\Delta(\alpha, \beta)$ can be computed in constant time,
where $\text{bit}_\Delta(\alpha, \beta)$ denotes the index of the first bit (after the radix point) at which $\alpha$ and $\beta$ differ when represented in binary.
Under this assumption, both $\mathcal{Q}$- and $\mathcal{Z}$-order comparisons can be performed in $O(1)$ time, using only the coordinates of the centers of cells for the $\mathcal{Q}$-order and point coordinates for the $\mathcal{Z}$-order.  
Similarly, for any two points in $P$, the smallest cell in the compressed quadtree that contains both can be computed in constant time using only their coordinates.  
Furthermore, given a point and a side length, we can compute in constant time the cell in the quadtree of that side length containing the point.

\subparagraph*{MPC construction the compressed quadtree.} We first generate all nodes of $T$, and then identify their parents and children.
We begin by sorting all points in $P$ according to the $\mathcal{Z}$-order.
This can be done efficiently in parallel, as $\mathcal{Z}$-order comparisons rely only on the coordinates of the points.
Next, for each pair of consecutive points under the $\mathcal{Z}$-order, we construct the smallest cell containing both points.
Each such cell corresponds to a non-leaf of the compressed quadtree.
Then, we have all nodes in the compressed quadtree since each point corresponds to a unique leaf.
Subsequently, we compute the parent and children for each node in the compressed quadtree.
In particular, for each $1 \leq i \leq 2^d$, we sort the cells, so that the cell corresponding to a node and that of its $i$-th child under the DFS order appear consecutively.
This can be done by modifying the pre-defined order among siblings, which defines the $\mathcal{Q}$-order.
Consequently, this allows us to associate each node with its $i$-th child in parallel.
We repeat this for all $1 \leq i \leq 2^d$ to associate every node with its children.

\subparagraph*{Constructing a \textsf{WSPD}.}
As for doubling metric spaces, we have the following observation from~\cite{har2011geometric}:
For a pair $(u,v)\in \mathcal{W}$ constructed by Algorithm~\ref{alg:har-wspd} such that the recursion on $\{u,v\}$ is invoked from the recursion on $\{u, \parent{v}\}$,
$\dist{}(\nodeCell{\parent{u}}, \nodeCell{\parent{v}}) \leq (1/\eps)\cdot \diam(\nodeCell{\parent{v}})$, and $\diam({\nodeCell{\parent{v}}})\leq \diam({\nodeCell{\parent{u}}})$.\footnote{To make the description consistent with~\cite{har2011geometric},
instead of using $r_v$ in Algorithm~\ref{alg:har-wspd},
we use $\diam(\cell_v)$ to measure the diameter of the point set corresponding to $v$, and we use $\dist{}(\cell_u,\cell_v)$ to measure the distance between two points from $\cell_u$ and $\cell_v$.}
Here, we refer to a node $u\in T$ as a pairing candidate of a node $v\in T$ if $u$ satisfies the above bounds.
Then, we identify all pairing candidates for each node in $T$.
To do this, we construct auxiliary cells for each node $v$ whose side length equals that of $\nodeCell{\parent{v}}$, which align with the hierarchical partition defining $T$ and lie within distance $(1/\eps)\cdot\diam(\nodeCell{\parent{v}})$ from $\nodeCell{\parent{v}}$.
We then sort all auxiliary cells and all nodes in the $\mathcal{Q}$-order of their corresponding cells, breaking ties by letting the auxiliary cell precede the node corresponding to the identical cell.
Observe that the successor node of an auxiliary cell of $v$ or its children may be a pairing candidate for the children of $v$.
Thus, we can identify all pairing candidates for each node in parallel.
It concludes Theorem~\ref{thm:intro-eclidean-algorithm}.

\begin{theorem}
    Given a set $P$ of $n$ points in $\mathbb R^d$ for a constant $d$, one can construct the compressed quadtree of $P$ of $O(n)$ size and complexity in $O(1)$ rounds, using $O(n)$ total space and $O(n^\delta)$ local space per machine for any $\delta \in (0,1)$.
\end{theorem}

\section{WSPD in Doubling Metrics} \label{sec:doubling}
In this section, we present an $O(1)$-round MPC algorithm to construct an $(1/\eps)$-\textsf{WSPD} of size $(1/\epsilon)^{O(\ddim)}n\log^2 n$ for a set $P$ of $n$ points in a metric space $\metric$ with a constant doubling dimension $\ddim$. Here, we use $(1/\eps)^{O(\ddim)} n\log^2 n/\localSize$ machines, each with local memory of size $O(\localSize)$, where $\localSize=n^{\delta}$ for any $\delta\in(0,1)$.
Our algorithm is randomized. While it always outputs an $(1/\eps)$-\textsf{WSPD} and guarantees that no overflow occurs, it may fail if the number of machines is insufficient to simulate the computation. The probability of such a failure is bounded by a constant.
We can amplify the success probability to $1-1/n^{\Omega(1)}$ by
increasing the number of machines to $(1/\eps)^{O(\ddim)} n\log^3n/\localSize$.

If a metric space $(X,\dist{})$ has doubling dimension $\ddim$,
any subspace $(P,\dist{}|_{P\times P})$ has doubling dimension at most $2\cdot \ddim$ for any subset $P$ of $X$. Therefore, in the following,
we assume that the given metric space is  $(P,\dist{}|_{P\times P})$ with $|P|=n$.

\subsection{MPC Algorithm for Computing a Partition Cover Tree}\label{sec:net-tree}
In this subsection, we show how to compute a partition cover tree in $O(1)$ rounds. 
Recall that $P$ is distributed across $O(n/\localSize)$ machines each with local memory of size $O(\localSize)$. 
As the output, we represent the tree $T$ as follows.
For each node $v$ of $T$, it is stored with a representative $p_v$ and a radius $r_v$,
where $r_v$ is the largest distance from $p_v$ to any point in $P_v$. The representative $p_v$ is chosen such that 
$p_v\subseteq P_u$ for the parent $u$ of any node $v$. 
The nodes of $T$ are stored in arbitrary order across the machines. 
For each node $v$, we store the associated data $(p_v,r_v)$, along with the IDs of the machines storing its parent and children in $T$.

\subparagraph*{Construction.}
We construct the partition cover tree $T$ in a top-down, \emph{layer-by-layer} manner.
We divide the $O(\log n)$ levels of $T$ into $(\log n/\log s')=O(1)$ \emph{layers}, each consisting of $\log s'$ consecutive levels, where $s' = 2^{O(\ddim)} \cdot \localSize^{1/2} \cdot \log^2 n\leq 2^{O(\ddim)}s=O(s)$.

\begin{figure}
    \centering
    \includegraphics[width=0.75\textwidth]{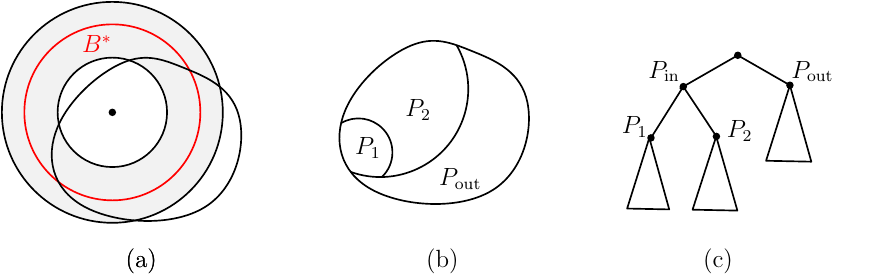}
        \caption{\small 
        (a) We choose a ball $B^*$ in the gray region uniformly at random.
        (b) If the radius of $B^*$ is large, we subdivide $P_\text{in}$ further into two subsets using a ball with small radius.
        (c) We have three nodes, and then construct the trees for each of the three nodes recursively using the same sample $Y$.
    \label{fig:partition-cover}}
\end{figure}

We first show how to compute the topmost layer. The implementation details will be provided in the proof of Lemma~\ref{lem:num-rounds}.  
Let $Y$ be a sample chosen by $s'$ independent draws from $P$. 
Note that $Y$ fits within the local memory of a single machine. 
Then using $Y$, we recursively construct the topmost layer as follows. Let $P'$ be the current piece where its ancestors are already constructed, and let $Y'=P'\cap Y$. Initially, $P'=P$ and $Y'=Y$. 
We compute the smallest-radius ball $B$ centered at a point of $Y'$ containing at least $|Y'|/4^{O(\ddim)}$ points from $Y'$.
Let $c$ and $r$ be the center and radius of $B$, respectively.
Note that the ball centered at $c$ with radius $2r$ 
contains at least $|Y'|/4^{O(\ddim)}$ and at most 
$|Y'|/2^{O(\ddim)}$ by the choice of $B$ and the doubling property.
 We choose a value $\theta$ uniformly at random from $[0,1]$, and let $B^*$ be the ball centered at $c$ with radius $(1+\theta)r$. See Figure~\ref{fig:partition-cover}(a).
 We create the node having $P'$ as its piece and having $B^*$ as its \emph{separating ball}. 
The separating ball partitions $P'$ into two subsets: $P_\text{in}=P'\cap B^*$ 
and $P_\text{out}=P'\setminus B^*$.
Let $\bar r$ be the radius of the separating ball of the parent of the current piece.
For the case that the current piece is $P$,
we let $\bar r$ be the diameter of $P$.
\begin{itemize}
    \item If $r(B^*) \leq \bar r/10$, we recursively construct trees for $P_\text{in}$ and $P_\text{out}$ using $Y \cap P_\text{in}$ and $Y \cap P_\text{out}$, respectively.
We then merge the results by attaching them to the node with piece $P'$. 
    \item Otherwise, we additionally construct one more level as follows. See Figure~\ref{fig:partition-cover}(b--c). 
    We begin by creating the node associated with $P'$, and its two children for $P_\text{in}$ and $P_\text{out}$.
    For the child $u$ with piece $P_\text{in}$, we select an arbitrary point $c'' \in P_\text{in} \cap Y$, and choose a random value $\theta \in [0,1]$.
    We then partition $P_\text{in}$ into two subsets using the ball $B(c'', (1+\theta)\bar r/20)$. These two subsets become the children of $u$. 
    
    Then for each of the nodes corresponding $P_\text{out}$ and the two subsets of $P_\text{in}$,
    we construct the tree by applying the entire procedure recursively using the same sample $Y$.
\end{itemize}
Here, the recursion terminates when the current piece contains at most $O(1)$ points from $Y$, in which case we simply create a leaf node for the set.

Let $T_1$ be the tree obtained in this way using the sample $Y$. 
Recall that in each node of $T_1$, we choose a ball separating the \emph{sampled points} in a balanced way. 
But we need to ensure that it also separates the \emph{input points} in a balanced way. This is proved in the following lemma.
For a node $v$ of $T_1$, let $P_v$ denote the piece of $v$.

\begin{lemma}\label{lem:tree-size}
    The size of $P_v$ is $O(n/s^{1/4})$ with high probability
    for all leaves $v$ of $T_1$.
\end{lemma}
\begin{proof}
Let $\mathcal R = \{ \cap_{B\in\mathcal B} B \setminus \cup_{B'\in \mathcal B'}B' \mid 
\mathcal B \text{ and } \mathcal B' \text{ are sets of at most $2^{O(\ddim)}(\log n)$ balls}
\}$. By construction, $P_v$ is contained in $\mathcal R$ for any leaf $v$ of $T_1$. 
Although the range space defined by balls in a doubling metric space does not necessarily have a bounded VC dimension~\cite{huang2018epsilon}, 
it has a small-sized $(1/s^{1/4})$-approximation as the number of ranges is small. By Lemma~\ref{thm:a-approx-sampling-poly}, $Y$ is 
a $(1/{s}^{1/4})$-approximation of $(P,\mathcal R)$ with high probability
since the number of ranges of $\mathcal R$ is $n^{2^{O(\ddim)}\log n}=n^{O(\log n)}$. Conditioned on the event that $Y$ is a 
 $(1/{s}^{1/4})$-approximation of $(P,\mathcal R)$, 
$P_v$ has $O(n/s^{1/4})$ points of $P$ for all leaves $v$ of $T_1$.
\end{proof}

In this way, we can compute the first layer of the partition cover tree of $P$. 
For the next layers, we recursively construct the tree starting from each leaf of $T_1$ until a leaf corresponds to a piece consisting of a single point. By Lemma~\ref{lem:tree-size}, the number of layers is $O(1)$ with high probability.

\begin{lemma}\label{lem:num-rounds}
    The partition cover tree can be computed in $O(1)$ rounds with high probability. 
\end{lemma}
\begin{proof}
It is sufficient to show that the first layer can be computed in $O(1)$ rounds as the number of layers is $O(1)$ with high probability. 
We can compute a sample $Y$ in $O(1)$ rounds~\cite{agarwal2016parallel}.
In a single machine, we first compute the hierarchical partition of $Y$. 
Each subset in the partition can be specified by its center, say $c$, and its radius, say $r$. Since they are determined by $Y$ only, we can compute it in a single machine. 
By broadcasting the centers and the radii to all machines, and 
we can compute the leaf of the first layer of the tree containing each point of $P$ in $O(1)$ rounds for all points in $P$.
Then we sort the points of $P$ with respect to the pieces of the leaves containing them in the lexicographical order in $O(1)$ rounds. In this way, we can place the points of $P$ contained in the same leaf in consecutive machines in $O(1)$ rounds.
Therefore, we can recursively apply the algorithm on each leaf of $T_1$.
\end{proof}


\subparagraph*{Analysis.}
We show that $T$ satisfies the small cover property.
Recall that $T$ is constructed randomly. To analyze the behavior of $T$, we represent $T$ as a subgraph of $T_\text{init}$,
where $T_\text{init}$ is 
the full binary tree with $n$ leaves, 
such that each node of $T$ corresponds to a node of $T_\text{init}$.
In the following, when we refer to a node of a random tree $T$, we indeed means its corresponding node of $T_\text{init}$.

Recall that a cover of a ball $B$ in the metric space induced by $T$ 
is a set $\{v_1,v_2,\ldots,v_k\}$  of nodes of $T$ such that 
the union of $P_{v_1},\ldots,P_{v_k}$ contains $B$,
and 
$\diam(P_{v_i})\leq \diam(B)/10$ for all nodes $v_i$.
 In particular, we can compute a cover of $B$ by traversing
$T$ from the root towards leaves. 
    Suppose that we are at a node $v$ 
    associated with a separating ball $B_v^*$. That is, the piece of a child of $v$ is $P_v\cap B_v^*$, and the piece of the other child of $v$ is $P_v\setminus B_v^*$. 
    We refer to the former as the \emph{in-child} of $v$,
    and the latter as the \emph{out-child} of $v$.
    We move to zero, one, or two of its children as follows. For each  child $w$ of $v$ with $P_w\cap B\neq\emptyset$, 
\begin{itemize} 
    \item we move to $w$ if it is the out-child of $v$, and 
    \item we move to $w$ if it is the in-child and
        $r(B_v^*)> \diam(B)/10$. 
\end{itemize}
At node $v$, we do not move to the in-child $w$ of $v$ if $P_w\cap B\neq\emptyset$ and $r(B_v^*) \leq \diam(B)/10$. Instead, we simply \emph{keep} this child as a node for the cover of $B$.
Note that the set of all kept nodes and all leaves we reach at the end of the traversal is a cover of $B$ induced by $T$. 
This is because $r(B_v^*)\geq \diam(P_w)$ for the in-child $w$ of $v$. 
We call the cover of $B$ constructed in this way the \emph{canonical cover} of $B$ induced by $T$.

We show that the expected size of the canonical cover of a fixed ball $B$ is $O(\log n)$. 
We say that $B_v^*$ \emph{separates} $B$ if $(B\cap P_v)\cap B_v^*\neq\emptyset$ and $(B  \cap P_v)\setminus B_v^*\neq\emptyset$.
Let $v$ be a node of $T$ with a separating ball $B_v^*$.
Then conditioned on visiting $v$ during the traversal, 
the probability that $B_v^*$ separates $B$ is at most $r(B)/r(B_v)$, where
$B_v$ is the smallest-radius ball containing at least $|Y|/4^{O(\ddim)}$ points. Recall that $r(B_v)\leq r(B_v^*)\leq 2r(B_v)$. 
Here, $r(B_v)$ is also a random variable that is determined by the random choices made for the ancestors of $v$, but it is independent to the random choice made for $v$.

The main challenge is to bound the number of occurrences in which the traversal moves to both children on a node, since such events increase the size of the canonical cover.
To handle this, we associate every visited node $v$ with a random variable $p_v$ which upper bounds the probability that $B^*_v$ separates $B$.
We first identify several cases where we always move to at most one child regardless of the random value chosen for a vertex $v$ visited during the traversal.
The first case is when $2r(B_v) \leq \diam(B)/10$; in this case, we move to the out-child only. Even if the piece of the in-child intersects $B$, we do not move to the in-child; instead, we simply keep this piece. 
The second case is when $B_v$ contains $B$; in this case, we either move to the in-child only, or do not move at all and simply keep the in-child. 
The third case is when $B\cap 2B_v =\emptyset$; in this case, we move to the out-child only,
where $2B_v$ denote the ball centered at $c(B_v)$ with radius $2r(B_v)$. 
We let $p_v=0$ if $v$ belongs to one of the three cases. Otherwise, we let $p_v=r(B)/r(B_v)$.
Just as $r(B_v)$ is a random variable that is defined by the random choices made for the ancestors of $v$, so is $p_v$.
Note that the conditional probability that we move to both children of $v$, given that $v$ is visited,  is upper bounded by the expected value of $p_v$ over the distribution conditioned on visiting $v$.

\begin{lemma}\label{lem:total-sum}
    Let $B$ be a fixed ball in the metric space, and $T$ be a fixed constructed tree. 
    For any root-leaf path $\pi$ of $T$, the sum of $p_u$ for all nodes $u$ in $\pi$ is $O(1)$.
\end{lemma}
\begin{proof}
    Let $v$ be an endpoint of $\pi$ which is a leaf of $T$. 
    For a node $u$, let $B_u^*$ be the separating ball of $u$,
    and let $B_u$ be the ball we have before the perturbation.
    For an ancestor $u$ of $v$,
    we call $u$ an \emph{in-ancestor} if $P_{v}$ is contained in $B_{u}^*$, and an \emph{out-ancestor} if $P_v$ lies outside of $B_{u}^*$. 
    Note that $P_v=\cap_{u\in \mathcal B} B_u^* \setminus (\cup_{B_{u'}\in\mathcal B'} B_{u'}^*)$, where 
    $\mathcal B$ denotes the set of the in-ancestors, and $\mathcal B'$ denotes the set of the out-ancestors.

    First, the sum of $p_u$ for all out-ancestors  
    is $2^{O(\ddim)}$. 
    To see this, we partition the ranges $[r,\infty]$ into ranges $I_0,I_1,\ldots$ where $I_i= [2^ir,2^{i+1}r]$ and $r=\diam(B)/20$. 
    For any out-ancestor $u'$ with $r(B_{u'})\in I_i$,
    we have $p_{u'}\leq 10/2^{i}$. 
    For any out-ancestor $u'$ with $r(B_{u'})\notin [r,\infty]$,
    we have $p_{u'}=0$. 
    We show that for each index $i$, there are $2^{O(\ddim)}$ out-ancestors $u'$ 
    with $r(B_{u'})\in I_i$ and $p_{u'}>0$. 
    Let $v_1,v_2,\ldots,v_k$ be the out-ancestors with $r(B_{v_j})\in I_i$ and $p_{v_j}>0$ for all indices $j$ sorted
    along $\pi$ from the root. Let $B_j=B_{v_j}$.
    Because $v_1,v_2,\ldots,v_{j-1}$ are out-ancestors of $v_j$, the center of $B_j$ is contained in 
    $P\setminus \{B_1\cup B_2\cup\ldots\cup B_{j-1}\}$.
    Thus no two centers of all such balls are within distance at most $2^ir$.
    On the other hand, since $p_{v_j}>0$, the distance
    between $c(B)$ and $c(B_j)$ is at most $2r(B_j)\leq 2^{i+1}r$ (see the third condition for $p_v=0$). 
    By the property of the doubling space,
    the number of such balls is $2^{O(\ddim)}$.
    Therefore, the sum of $p_u$ for all out-ancestors is $20\cdot 2^{O(\ddim)}=O(1)$.
    
    Second, the sum of $p_u$ for all in-ancestors is at most a constant.
    To see this, we consider the following two cases. 
    In the first case, consider
    all in-ancestors $p_u$ such that $p_u$ increases by at least a constant factor from its parent.
    Since $p_u$ is always at most {$10$}, the sum of $p_u$ for all such nodes $u$ is at most a constant.
    
    Now consider the remaining in-ancestors $u$. Fix an arbitrary subpath $\pi'$ of $\pi$ consisting of in-ancestors of length, say $k\geq 2$.
    By construction, for at least $\lfloor k/2\rfloor$ in-ancestors $p_u$ in $\pi'$, $p_u$ increases by at least constant factor from its parent. As this holds for any such path of length at least two, the sum of $p_u$ in this case equals the sum of $p_u$ for all in-ancestors in the above case and all out-ancestors, and thus it is also $O(1)$.
\end{proof}

\begin{lemma}\label{lem:expected size}
     For any ball $B$, the expected size of the canonical cover of $B$ induced by $T$ is $O(\log n)$.
\end{lemma}
\begin{proof}
    For a node $v$ of $T_\text{init}$, let 
    $X_v$ denote the number of nodes visited by the traversal 
    in the subtree of $T$ rooted at $v$. 
    If $v$ does not appear in $T$, let $X_v=0$.  
    In the following, we let $\mathbb{E}_v[X_v]$ denote
    $\mathbb{E}[X_v \mid v \text{ is visited}]$.
    Then we have
    \[
    \mathbb{E}_v[X_v] \leq \mathbb{E}_\ell[X_\ell] \cdot (\hat p_v+ \hat p_{\ell,r}) + \mathbb{E}_r[X_r]\cdot (\hat p_v+ \hat p_{r,\ell})+1,
    \]
    where $\ell$ and $r$ are the children of $v$,
    and $\hat p_{a,b}$ is the probability that we visit $a$ but not $b$ 
    conditioned on visiting $v$,
    and $\hat p_v$ is the probability that we visit both $\ell$ and $r$ conditioned on visiting $v$.
    Note that $\hat p_{r,\ell}+\hat p_{\ell,r} \leq 1-\hat p_v$.
  For a node $v$, let $\hat A_v = \max_{\pi} \sum_{w \in \pi} \hat p_w$
and $A_v = \max_{\pi} \sum_{w \in \pi} p_w$, where the maximum is taken over all paths $\pi$ from $v$ to its descendants.
Note that $\hat A_v$ is a scalar, and $A_v$ is a random variable.

To show that $\hat A_v = O(1)$, we fix a path $\pi$ from $v$ to a leaf and show that the sum of $\hat p_w$ over all $w \in \pi$ is $O(1)$.
For each node $w \in \pi$, let $\pi[w]$ denote the subpath from $v$ to $w$, i.e., the ancestors of $w$ (including $w$) in $\pi$. 
Define $A_{\pi,w} = \sum_{w' \in \pi[w]} p_{w'}$.
Let $\mathcal{E}_w$ be the event that all ancestors and $w$ are visited, and no descendant of $w$ in $\pi$ is visited.
Recall that $\hat p_w \leq \mathbb{E}[p_w \mid \text{visited }w]$.
Then we have 
\begin{align*}
    \sum_{w \in \pi} \hat{p}_w 
    &\leq \sum_{w \in \pi} \mathbb{E}[p_w \mid \text{visited }w]  
    \leq \sum_{w \in \pi} \mathbb{E}[ A_{\pi,w} \mid \mathcal{E}_w] \\
    &\leq \sum_{w \in \pi} A_{\pi,w} \cdot \Pr[\mathcal{E}_w]
    \leq O(1) \cdot \sum_{w \in \pi} \Pr[\mathcal{E}_w] \\
    &\leq O(1)
\end{align*}
The second inequality holds by decomposing the event “visited $v$” into the finer disjoint events $\mathcal{E}_{w}$'s.
The third inequality follows from the definition of conditional expectation.
The fourth inequality uses Lemma~\ref{lem:total-sum}, which states that $A_{\pi,w} = O(1)$ if $\prob[\mathcal E_w]>0$.
Finally, the fifth inequality holds since the events $\mathcal{E}_w$ are disjoint, and thus the total probability is at most 1.
    
    Now we claim that $\mathbb{E}_v[X_v] \leq h_v\cdot e^{\hat A_v}$, where $h_v$ is the height of the subtree of $T_\text{init}$ rooted at $v$.
    We prove this using induction on depth. If $v$ is a leaf, $X_v\leq 1$, and thus we are done.
    If $v$ is not a leaf, then $\hat A_\ell\leq \hat A_v- \hat p_v$ and $\hat A_r\leq \hat A_v-\hat p_v$ by definition.
    Then we have
   \begin{align*}
    \mathbb{E}_v[X_v] &\leq (h_v-1)\cdot e^{\hat A_v -\hat p_v} \cdot (2\hat p_v+\hat p_{\ell,r}+\hat p_{r,\ell})+1\\
    &\leq (h_v-1)\cdot e^{\hat A_v -\hat p_v} \cdot (1+\hat p_v)+1\\
    &\leq (h_v-1)\cdot e^{\hat A_v -\hat p_v}\cdot 2^{\hat p_v}+1\\
    &\leq h_v \cdot e^{\hat A_v}.
    \end{align*}

    Since the height of $T$ is always $O(\log n)$ and $\hat A_v=O(1)$, the expected size of the canonical cover of $B$ induced by $T$ is at most $\mathbb{E}[X_\textsf{root}]=\mathbb{E}_\textsf{root}[X_\textsf{root}]=O(\log n)$, where $\textsf{root}$ is the root of $T_\text{init}$.
    Here, we can get rid of the condition as we always visit the root of $T$.
\end{proof}

\subsection{MPC Algorithm for Computing a WSPD}
In this subsection, we show how to compute a \textsf{WSPD} in $O(1)$ rounds assuming that we have a partition cover tree $T$.
We first present an algorithm that computes a \textsf{WSPD} of size $(1/\epsilon)^{O(\ddim)}n\log n$ with a constant success probability. Then we show that by constructing $O(\log n)$ different partition trees in parallel,
we can amplify the success probability to $1-1/n^{O(1)}$. 
This increases the sizes of the number of machines and the \textsf{WSPD} by $O(\log n)$.
For a node $v$ of $T$, we use $p_v$ to denote the representative of $P_v$, and $r_v$ to denote the maximum distance from $p_v$ to a point in $P_v$. Each node $v$ of $T$
is stored along with $p_v$ and $r_v$.

\subparagraph*{Construction.}
Our strategy is to compute 
all \emph{candidate pairs} in parallel, and we return the well-separated ones among all candidate pairs.
For a node $u$ of $T$,
its \emph{pairing candidates} in $T$ are the nodes $v$ in $T$ 
such that $\dist{}(p_v,p_{\parent{u}})\leq (8/\varepsilon)r_{\parent{u}}$, and $r_{\parent{u}}\leq r_{\parent{v}}$.
In this case, we call $(u,v)$ is a \emph{candidate pair} (for \textsf{WSPD}).\footnote{Although we call it a \emph{pair}, we treat it as a set. That is, we do not distinguish between $(u,v)$ and $(v,u)$.}

Our algorithm consists of $O(1)$ \emph{blocks} of rounds.
For each block, we start with a set $\mathcal{F}$ of pairs of subtrees of $T$. 
Recall that the levels of $T$ are partitioned into $O(1)$ \emph{layers}, each consisting of $O(\log \localSize^{1/4})$ consecutive levels.
Let $\mathcal T$ be the set of all subtrees $T'$ of $T$ such that
the root of $T'$ lies on the first level of a layer, and 
the leaves of $T'$ lie on the first level of the next layer. Here, each vertex in the first level of each layer appears in two subtrees of $\mathcal T$.
Note that every subtree of $\mathcal T$ consists of $O(\localSize^{1/4})$ nodes. For a node $v$ of $T$, let $T_v$ be the subtree in $\mathcal T$ containing $v$ as an internal node.
The first block starts with $\{(T_\textsf{root},T_\textsf{root})\}$ and $\{(\textsf{root},\textsf{root})\}$, where $\textsf{root}$ is the root of $T$. 

Now suppose that we are at a block with $\mathcal F$.
We assign one machine for each pair $(T_1,T_2)$ of $\mathcal F$ and store $T_1$ and $T_2$ on this machine in $O(1)$ rounds. Then we want to compute all candidate pairs consisting of nodes in $T_1$ and nodes in $T_2$. 
However,
it is possible that the number of such pairs exceeds the size of the local memory. Thus we first compute the number of such candidate pairs, and then we assign a sufficient number of machines to $(T_1,T_2)$. Notice that computing such candidate pairs requires memory size linear in the output, but computing the number of such candidates requires the size linear in the complexity of $T_1$ and $T_2$. 
We broadcast $T_1$ and $T_2$ to those machines, 
and the $i$th machine processes the pairs whose ranks fall in the interval $[is,(i+s)]$ 
where the ranks are determined in lexicographical order. 
By broadcasting the IDs of the first and last machines processing $(T_1,T_2)$, each machine processing $(T_1,T_2)$ can recognize its index $i$. 
Then we broadcast the well-separated candidate pairs to the predetermined machines that collect the pairs of the \textsf{WSPD}, and we pass the candidate pairs $(u,v)$ obtained from $(T_1,T_2)$ that are not well-separated, and either $u$ or $v$ is a leaf node of $T_1$ or $T_2$.
In particular, we produce $(T_u,T_v)$ for the tree pair set for the next block.
After producing all such tree pairs (represented by the indices of the trees), 
we remove duplicated pairs, and pass the resulting set to the next block.


\begin{lemma}
    This algorithm runs in $O(1)$ rounds.
\end{lemma}
\begin{proof}
    For a pair $(T_1,T_2)$
    of trees of $\mathcal T$, we define its \emph{level} as the sum of the indices of the layers of $T_1$ and $T_2$ in $T$. 
    In the first block, we start with $\{(T_\textsf{root},T_\textsf{root})\}$, and the level of the tree pair is two.
    Then in each subsequent block, the smallest level of the pairs we are given increases by at least one. Any pair of trees of $\mathcal T$ has level at most $O(1)$. Therefore, the number of blocks is $O(1)$, and thus the number of rounds is also $O(1)$.
\end{proof}

\subparagraph*{Correctness analysis.}
To show that the returned pairs form an $(1/\varepsilon)$-\textsf{WSPD}, we consider the set $\mathcal W^*$ of pairs computed by applying Algorithm~\ref{alg:har-wspd} on $T$. 
Also, let ${\mathcal W}^*_\textsf{all}$ be the set of all pairs generated by applying Algorithm~\ref{alg:har-wspd} on $T$.
Note that $\mathcal W^* \subseteq {\mathcal W}^*_\textsf{all}$.
Clearly, $\mathcal W^*$ is an $(1/\varepsilon)$-\textsf{WSPD}. 
We first show that all node pairs of $T$ generated by Algorithm~\ref{alg:har-wspd} are candidate pairs, and then we show  that all candidate pairs are generated by our MPC algorithm.

\begin{lemma}\label{lem:returned pair is a candidate pair}
    Every pair $(v,u)$ of ${\mathcal W}^*_\textsf{all}$ is a candidate pair.
\end{lemma}
\begin{proof}
    Assume first that $(\parent{v},u)$ generates $(v,u)$ during the execution of Algorithm~\ref{alg:har-wspd}. Then $r_{\parent{u}}\geq r_{\parent{v}} \geq r_{u}$. 
    Since $(\parent{v},u)$ is not well-separated, the distance between their representatives is less than $(8/\varepsilon)r_{\parent{v}}$. 
    Therefore, $(u,v)$ is a candidate pair.
    For the other case, we change the roles of $u$ and $v$,
    which shows that $(v,u)$ is a candidate pair.
    As we do not distinguish between $(u,v)$ and $(v,u)$,
    the lemma holds.
\end{proof}

\begin{lemma}
    Our MPC algorithm returns an $(1/\varepsilon)$-\textsf{WSPD}.
\end{lemma}
\begin{proof}
    All pairs of $\mathcal W$ are well-separated by construction.
    Thus it is sufficient to show that every pair of points of $P$ is covered by at least one pair of $\mathcal W$.
    For this, it is sufficient to show that
    $\mathcal W$ is a superset of $\mathcal W^*$.

    For a pair $(u,v)$ of $\mathcal W^*_\textsf{all}$, 
    consider the sequence $\mathcal P'=\langle (u_1,v_1),(u_2,v_2),\ldots,(u_k,v_k) \rangle$ of pairs in ${\mathcal W}^*_\textsf{all}$ such that $(u_1,v_1)=(\textsf{root},\textsf{root})$, $(u_k,v_k)=(u,v)$, and $(u_i,v_i)$ generates $(u_{i+1},v_{i+1})$ for all $i$.
    We claim that all pairs in the sequence are generated by our MPC algorithm. Clearly, the first pair $(u_1,u_2)$ is generated by the first block of the algorithm.
    Assume that $(u_i,v_i)$ is generated by a tree pair $(T_i,T_i')$. 
    If $u_{i+1}\in T_i$ and $v_{i+1}\in T_i'$, then  $(u_{i+1},v_{i+1})$ is also generated. 
    Otherwise, $u_i$ is a leaf of $T_i$, or $v_i$ is a leaf of $T_i'$.
    In this case, we consider $(T_{u_i},T_i')$ or $(T_i,T_{v_i})$
    in the next block, and $(u_{i+1},v_{i+1})$ is generated in the next block. Recall that each vertex of $T$ in the first level of each layer appears two subtrees of $\mathcal T$. For a node $v$ of $T$, let $T_v$ be the subtree in $\mathcal T$ containing $v$ as an internal node. 
\end{proof}

\subparagraph*{Space complexity analysis.}
We first analyze the expected number of pairs in the \textsf{WSPD},
and the expected total memory size.
As a single machine always uses $O(s)$ space, this gives the expected number of machines required by the algorithm.
In the following, we assume the event that all random samples used for the $\alpha$-approximations in the construction of the partition cover tree are indeed all $\alpha$-approximations.
Since each random sample succeeds with high probability, the probability that this event occurs is also at least $1-1/n^{\Omega(1)}$.

\begin{lemma}\label{lem:size-candidate}
    The expected number of pairing candidates of a node of $T$ is $(1/\eps)^{O(\ddim)}\log^2 n$.
\end{lemma}
\begin{proof}
    Consider a node $u$ of $T$.
    Let $\mathcal C_i$ be a set of pairing candidates of $u$ at level $i$. 
    As the height of $T$ is $O(\log n)$, it is sufficient to show
    that the expected size of $\mathcal C_i$ is $(1/\eps)^{O(\ddim)}\log n$ for all $i$. 
    
    Recall that for a node $u$ of $T$, its pairing candidates are the nodes $v$ with $\dist{}(p_v,p_{\parent{u}})\leq (8/\varepsilon)r_{\parent{u}}$, and $r_{v}\leq r_{\parent{u}}\leq r_{\parent{v}}$.
    Consider the ball $B$ centered at $p_{\parent{u}}$ with radius $(16/\varepsilon)r_{\parent{u}}$.
    Note that it contains a point of $P_{\parent{v}}$ since $p_v$ is also contained in  $P_{\parent{v}}$.
    By the packing property of a doubling metric space, $B$ can be covered 
    by  $(1/\varepsilon)^{O(\ddim)}$ balls of radii $r_{\parent{u}}$.
    Each such ball can be covered by $O(\log n)$ pieces induced by $T$ of radius at most $r_{\parent{u}}/10$ in expectation by Lemma~\ref{lem:expected size}.
    In this way, we have $(1/\varepsilon)^{O(\ddim)}\log n$ pieces  of radii  $r_{\parent{u}}/10$ whose union contains a point of $P_{\parent{u}}$. Let $\mathcal P$ be the set of such pieces.
    Since $r_{\parent{v}}$ is larger than any piece of $\mathcal P$, at least one of the pieces of $\mathcal P$ is fully contained in $P_{\parent{v}}$. 
    Since all pieces of the nodes of $\mathcal C_i$ are  pairwise disjoint, and each node has two children, the size of $\mathcal C_i$ is $O(|\mathcal P|)=(1/\varepsilon)^{O(\ddim)}\log n$. Therefore, the lemma holds.
\end{proof}

By Lemma~\ref{lem:returned pair is a candidate pair} and Lemma~\ref{lem:size-candidate}, we have the following lemma
since $T$ has $O(n)$ nodes.
\begin{lemma}\label{lem:size WSPD}
    The expected total space and the expected size of $\mathcal W$ are $(1/\eps)^{O(\ddim)}\log^2 n$.
\end{lemma}

\subparagraph*{Success probability.}
Our algorithm computes a $(1/\varepsilon)$-\textsf{WSPD} of a point set of size $n$ in $O(1)$ rounds.
This algorithm is randomized; the expected total memory size required by the algorithm is $(1/\varepsilon)^{O(\ddim)}n\log^2n$, and the expected size of the returned \textsf{WSPD} is $(1/\varepsilon)^{O(\ddim)}n\log^2n$.
Due to Markov's inequality, we can see that 
the probability that the total memory size is $(1/\varepsilon)^{O(\ddim)}n\log^2 n$, and the size of the \textsf{WSPD} is $(1/\varepsilon)^{O(\ddim)}n\log^2 n$ is at most a constant.
As each machine always uses $O(s)$ space, the expected number of machines required by the algorithm is $(1/\varepsilon)^{O(\ddim)}(n\log^2n)/s$.

\begin{lemma}\label{lem:ddim-wspd-const-prob}
    Given a set $P$ of $n$ points in a metric space with a constant doubling dimension $\ddim$,
    we can compute a $(1/\varepsilon)$-\textsf{WSPD} of $P$ of size $(1/\eps)^{O(\ddim)}n\log^2 n$
    in $O(1)$ rounds, using $(1/\eps)^{O(\ddim)}(n\log^2 n)/s$ machines, each with $O(s)$ space. The algorithm succeeds with at least constant probability.
\end{lemma}
 
We can amplify the success probability to $1-1/n^{\Omega(1)}$ by building $O(\log n)$ independent partition cover trees in parallel.
We partition the machines into $O(\log n)$ blocks, and each block
runs the algorithm independently in parallel. If one block requires more machines than allocated, we terminate its execution.
At the end of the algorithm, we return the output produced by any block that successfully completed.  Then we have the following lemma.

\begin{lemma}\label{lem:ddim-wspd-high-prob}
    Given a set $P$ of $n$ points in a metric space with a constant doubling dimension $\ddim$,
    we can compute a $(1/\varepsilon)$-\textsf{WSPD} of $P$ of size $(1/\eps)^{O(\ddim)}n\log^2 n$
    in $O(1)$ rounds, using $(1/\eps)^{O(\ddim)}(n\log^3 n)/s$ machines, each with $O(s)$ space. The algorithm succeeds with $1-1/n^{\Omega(1)}$ probability.
\end{lemma}

\section{WSPD in Euclidean Space} \label{sec:euclid}
In this section, we present a fully scalable $O(1)$-round MPC algorithm to construct an $(1/\eps)$-\textsf{WSPD} of a set $P$ of $n$ points in the Euclidean space $\mathbb R^d$ for a constant $d$.
Our algorithm uses $(1/\varepsilon)^{O(d)}n$ total space and $O(n^\delta)$ space per machine for any $\delta\in (0,1)$.

\subsection{MPC Algorithm for Constructing a Compressed Quadtree} \label{sec:mpc-euclid-qt}
In this subsection, we describe our fully scalable $O(1)$-round MPC algorithm for constructing a compressed quadtree $\qt$ of $P$.
Here, we use $O(n)$ total space and $O(s)$ local space per machine, where $\localSize=n^\delta$ for any $\delta\in(0,1)$.
As the output, we represent the compressed quadtree $\qt$ as follows.
We store the nodes of $\qt$ in the DFS order, i.e., in the $\mathcal{Q}$-order of their corresponding cells.
Here, we associate each node $v$ with a representative point in $P$ lying in $\nodeCell{v}$, the diameter of $\nodeCell{v}$, and the IDs of the machines storing the parent and the children of $v$ in $\qt$.

As mentioned earlier, we first compute all cells corresponding to the nodes in the compressed quadtree, using its structural relation to the $\mathcal{Z}$-order of $P$.
Then, we compute the structural representation associated with each node by slightly modifying the $\mathcal{Q}$-order.

\subparagraph*{Finding all cells.}
We first sort all points in $P$ according to the $\mathcal{Z}$-order.
Then, for two consecutive points $p,q\in P$, we construct the smallest cell $\cell_{p,q}$ that contains both, and associate $\cell_{p,q}$ with one of $p$ and $q$.
Observe that we may construct multiple duplicates for some cells.
To ensure a canonical representation of the compressed quadtree, we sort all cells according to the $\mathcal{Q}$-order to remove duplicates.
Specifically, we assign an arbitrary rank to the duplicates for each cell, and then we remove the duplicates whose rank is not the smallest.

\begin{lemma}\label{lem:qt-cells}
    One can construct the cells corresponding to all nodes of $\qt$ in $O(1)$ rounds.
    Moreover, each node $v$ in $\qt$ is associated with its representative point contained in $\nodeCell{v}$.
\end{lemma}
\begin{proof}  
    We first show that we construct the cells corresponding to all nodes of the compressed quadtree $\qt$.  
    It suffices to show that we obtain all cells corresponding to internal nodes of $\qt$, since each leaf corresponds to a point in $P$ itself.
    Consider any internal node $u \in \qt$. Let $u_1$ and $u_2$ be two of its children such that $u_1$ appears before $u_2$ in the DFS order, and no other child of $u$ appears between them.
    Let $p$ be the last point in $\nodeCell{u_1}$ and $q$ the first point in $\nodeCell{u_2}$ under the $\mathcal{Z}$-order.  
    Since the DFS order induces the $\mathcal{Z}$-order, $p$ and $q$ are consecutive in the $\mathcal{Z}$-order, and we have $\nodeCell{u} = \cell_{p,q}$.
    
    We now show that all such cells can be constructed in $O(1)$ rounds.  
    Since the $\mathcal{Z}$-order comparison between two points depends only on the coordinates of the points, we can sort all points in $P$ according to this order in $O(1)$ rounds.  
    Let $p$ and $q$ be a pair of consecutive points in the sorted order.
    If $p$ and $q$ are stored on the same machine, the machine constructs $\cell_{p,q}$ and associates it with $p$.  
    Otherwise, the machine storing $p$ sends it to the machine storing $q$.
    Then the machine storing $q$ constructs $\cell_{p,q}$ and associates it with $q$.

    We can also remove the duplicates of each cell in $O(1)$ rounds.
    More specifically, we can sort all cells in $O(1)$ rounds, since each machine constructs at most $O(\localSize)$ cells and the total number of cells is $O(n)$.
    Furthermore, we can assign the indices to all duplicates of each cell in $O(1)$ rounds, and then each machine can remove the stored cells whose rank is not the smallest.
\end{proof}

By Lemma~\ref{lem:qt-cells}, we have the cells of all nodes of $\qt$.
From now on, we consider each cell as a node of $\qt$.
In particular, each point in $P$ is a leaf of $\qt$.

\subparagraph*{Computing the structural representation.}
Now, for each node $v$ of $\qt$, we compute the IDs of the machines storing its parent and children, and associate $v$ with them.
To this end, we first define the \emph{$\mathcal{Q}(i)$-order}, a total order derived from the $\mathcal{Q}$-order by modifying the underlying DFS order.
Specifically, for $1 \leq i \leq 2^d$, the $\mathcal{Q}(i)$-order is defined by the DFS traversal of $\qt$ in which the $i$-th child of each node is visited before its siblings.
See Figure~\ref{fig:Qi-order}.
Notice that the $\mathcal{Q}(i)$-order comparison for any $i$ can be performed in constant time using only the coordinates of the vertices of the cells; see~\cite{chan2020locality}.

\begin{figure}
    \centering
    \includegraphics[width=0.8\textwidth]{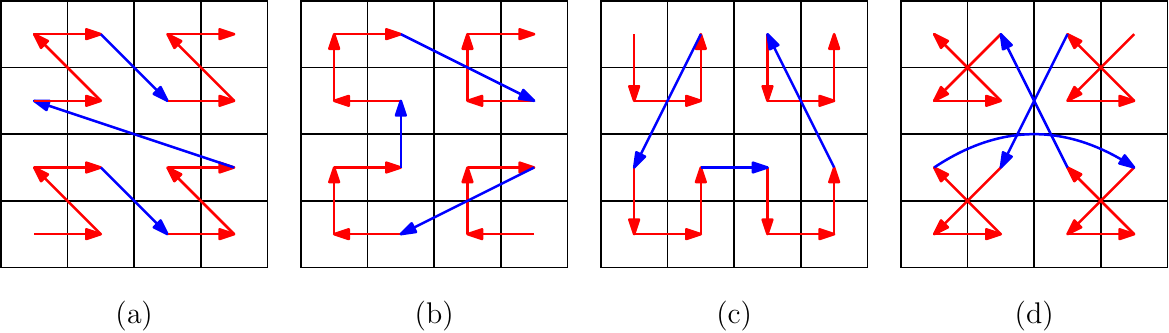}
        \caption{\small (a) Illustration of the $\mathcal{Q}$-order (or $\mathcal{Q}(1)$-order).
        (b)–(d) Illustration of the $\mathcal{Q}(2)$-, $\mathcal{Q}(3)$-, and $\mathcal{Q}(4)$-orders, respectively.
    \label{fig:Qi-order}}
\end{figure}

We first associate each node $v$ with the ID of the machine containing $v$.
Recall that all nodes of $\qt$ are distributed across the machines according to the DFS order, i.e., the $\mathcal{Q}$-order of their corresponding cells.

Next, we then iteratively sort all nodes in $\qt$ according to the $\mathcal{Q}(i)$-order of their corresponding cells, for $1 < i \leq 2^d$.
Let $u$ and $v$ be two consecutive nodes in the $\mathcal{Q}(i)$-order of their corresponding cells such that $u$ precedes $v$.
If $\nodeCell{u}$ contains $\nodeCell{v}$, then $u$ is the parent of $v$ in $\qt$, and by definition, $v$ is the $i$-th child of $u$.
We then associate $u$ with the machine ID of $v$, and associate $v$ with the machine ID of $u$.
Note that if a node $u$ has fewer than $2^d$ children, some iterations may attempt to associate $u$ with the children already considered.
To prevent this redundancy, we check whether each child of $u$ has already been considered before the association.

Finally, we sort all nodes of $\qt$ in the $\mathcal{Q}$-order of their corresponding cells to ensure that the machine IDs associated with each node correctly refer to their parent and children.

\begin{lemma}\label{lem:qt-association}
    One can compute the representation of all nodes of $\qt$, including the machine IDs of its parent and its children, in $O(1)$ rounds.
\end{lemma}
\begin{proof}
    For each $i \in [2^d]$, we sort all nodes of $\qt$ according to the $\mathcal{Q}(i)$-order of their corresponding cells.
    This sorting takes $O(1)$ rounds, as the total number of nodes is $O(n)$ and each machine handles at most $O(\localSize)$ nodes.    
    Once sorted, we check in parallel each pair of consecutive nodes $u$ and $v$ in the $\mathcal{Q}(i)$-order to determine whether $u$ is the parent of $v$.
    Specifically, if $u$ and $v$ are stored on the same machine, the machine checks whether $u$ is the parent of $v$ and, if so, associates $v$ with the machine ID of $u$ and $v$ with the machine ID of $v$.
    Otherwise, the machine storing $u$ sends it to the next machine, i.e., the machine storing $v$, which then performs the check and the association.
    Note that we can also check whether $v$ has already been considered as a child of $u$ using only local computation.
    Since each iteration takes $O(1)$ rounds, the total round complexity is $O(2^d) = O(1)$.
\end{proof}

Combining Lemma~\ref{lem:qt-cells} and Lemma~\ref{lem:qt-association}, we conclude the following theorem.

\begin{theorem}\label{thm:construct-qt}
    Given a set $P$ of $n$ points in $\mathbb R^d$ for a constant $d$, one can construct the compressed quadtree of $P$ of $O(n)$ size and complexity in $O(1)$ rounds, using $O(n)$ total space and $O(n^\delta)$ local space per machine for any $\delta \in (0,1)$.
\end{theorem}

\subsection{MPC Algorithm for WSPD in the Euclidean Space} \label{sec:mpc-euclid-wspd}
In this subsection, we describe a fully scalable $O(1)$-round MPC algorithm for constructing an $(1/\eps)$-\textsf{WSPD} of $P$ of size $(1/\varepsilon)^{O(d)}n$ for $\varepsilon > 0$.
Here, we use $(1/\varepsilon)^{O(d)} n$ total space and $O(n^\delta)$ local space per machine for $\delta \in (0,1)$.
As the output, each pair in the \textsf{WSPD} consists of two nodes of the compressed quadtree, each represented by the diameter of its corresponding cell and a representative point in $P$ contained in that cell.

Suppose that we have the compressed quadtree $\qt$.
As mentioned earlier in Section~\ref{sec:overview-Euclidean}, we first find all pairing candidates for each node $v\in \qt$ by constructing auxiliary cells for every node.
Let $\grid{\ell}$ denote the subdivision of $\mathbb R^d$ into hypercubes of diameter $\ell$.
For each node $v$ of $\qt$, we consider all cells $\cell$ in $\grid{\diam(\nodeCell{\parent{v}})}$ that satisfy $\dist{}(\cell, \nodeCell{\parent{v}}) \leq (1/\eps)\cdot\diam(\nodeCell{\parent{v}})$.
Then we associate each auxiliary cell with $v$. Then for any pairing candidate $u$ of $v$, its parent cell contains an auxiliary cell $\cell$ constructed for $v$.
More specifically, 
if $\cell$ corresponds to a node of $\qt$, then $\cell=\nodeCell{u}$ or $\cell=\nodeCell{\parent{u}}$.
Otherwise, $u$ is a unique node such that $\nodeCell{u}$ is the largest cell contained in $\cell$ among all cells corresponding to a node in $\qt$.

Next, we determine the pairing candidates for each node of $\qt$ using its auxiliary cells.
We first sort all auxiliary pairs and all nodes in $\qt$ in the $\mathcal{Q}$-order of their corresponding cells.
In case of a tie between an auxiliary cell and an identical cell corresponding to a node in $\qt$, the auxiliary cell precedes the other.
After that, for each auxiliary cell, we compute its successor \emph{node}. 
Then for each auxiliary cell $\Box$ of $v$, its successor node $u$ and its children are pairing candidates of every child of $v$ if $\Box$ contains $\Box_u$. If it is not the case, $\Box$ does not contain any point of $P$, and thus we simply ignore it.
Due to the previous observation,
we can compute all pairing candidates of each node 
in this way. 

Note that this is possible since every auxiliary cell corresponds to a cell in the hierarchical partition defining the underlying (regular) quadtree of $\qt$, and the $\mathcal{Q}$-order is preserved under compression.
After that, we construct the $(1/\eps)$-\textsf{WSPD} of $P$ by checking whether each node and its pairing candidates satisfy Inequality~\ref{equ:cell-separated}.

    
    


\begin{lemma}\label{lem:wspd-Euclidean}
    Given the compressed quadtree $\qt$ of $P$ and $\varepsilon>0$, one can construct a $(1/\eps)$-well-separated pair decomposition $\wspd$ of $P$ of size $(1/\varepsilon)^{O(d)} n$.
    Each pair of $\wspd$ corresponds to two nodes in $\qt$, respectively.
    It takes $O(1)$ rounds and uses $(1/\varepsilon)^{O(d)} n$ total space and $O(n^\delta)$ local space per machine for any $\delta \in (0,1)$.
\end{lemma}
\begin{proof}
    We first construct auxiliary cells for each node in parallel.
    Let $\corrM{v}$ be the machine storing $v$.
    Since we construct auxiliary cells of diameter the same as $\parent{v}$ for every node $v$, we first associate $v$ with $\nodeCell{\parent{v}}$.
    This can be possible in $O(1)$ rounds, since every node $v$ is associated with the diameter of its corresponding cells and the ID of its children.
    Specifically, for each node $v$, $\corrM{v}$ sends the diameter of its corresponding cell to its children.
    
    Then, $\corrM{v}$ constructs the auxiliary cells of $v$, and associates each of them with $v$.
    However, it is possible that the number of auxiliary cells constructed by one machine can exceed the size of the local memory, since the number of auxiliary cells of each node is $(1/\eps)^{O(d)}$.
    Thus, we use the additional machines for constructing the auxiliary cell of a node.
    In particular, $M_v$ broadcasts $v$ to the $i$th additional machine, which then constructs $[is,(i+s)]$-th auxiliary cells for $v$ in lexicographical order.
    Notice that $M_v$ can determine the ID of the additional machines for $v$, since all nodes are sorted in the DFS order, and the number of auxiliary cells for a node is the same for every node.
    This is possible since the number of all auxiliary cells for all nodes is $(1/\eps)^{O(d)}n$.
    Moreover, it takes $O(1)$ rounds.

    Next, we sort all auxiliary cells along with all nodes in $\qt$ in $O(1)$ rounds.
    With the same round complexity, we can find the successor node $v'$ of each auxiliary cell $\hat \cell_v$ of each node $v$ and store it in the machine storing $\hat \cell_v$ in parallel, by slightly modifying the \textsf{Predecessor} algorithm.
    Therefore, the machine storing $\hat \cell_v$ can determine if $v'$ is a pairing candidate of $v$, and associate $\hat \cell_v$ with $v'$ and the children of $v'$ accordingly.
    Notice that each machine may receive $(1/\eps)^{O(d)}\cdot \localSize$ messages if we send $\hat \cell_v$ to the children of $v'$ to associate them with $\hat \cell_v$.
    To avoid this, we associate every node with its children in $O(1)$ rounds before sorting.
    
    After that, each machine has a set of pairs, each consisting of a node and its pairing candidate.
    Recall that the auxiliary cells for a node are associated with that node.
    Each machine then determines which of these should be added to the \textsf{WSPD}.
    More specifically, for each pair $(v,v')$ where $v'$ is a pairing candidate of $v$, $M_{(v,v')}$ checks whether $v$ and $v'$ satisfies $\max\{\diam(\nodeCell{v}), \diam(\nodeCell{v'})\} \leq \varepsilon \cdot \dist{}(\nodeCell{v}, \nodeCell{v'})$. 
    If so, $M_v$ constructs a \textsf{WSPD} pair consisting of the representative points of $v$ and $v'$, along with the diameters of $\nodeCell{v}$ and $\nodeCell{v'}$.
    In particular, the diameter of $\nodeCell{v}$ is zero if $v$ is a leaf, i.e., $\nodeCell{v}$ is a point; the same holds for $v'$.
    
    By~\cite{har2011geometric}, it is immediate that the number of returned pairs is $(1/\varepsilon)^{O(d)} n$, and the pairs can be stored across the machines since each pair is represented by two cells.
\end{proof}

By combining Theorem~\ref{thm:construct-qt} and Lemma~\ref{lem:wspd-Euclidean}, we obtain the following result.

\begin{theorem}\label{thm:complete-wspd-euclid}
    Given a point set $P$ of size $n$ in $\mathbb R^d$ for a constant $d$ and $\varepsilon>0$, one can construct a $(1/\varepsilon)$-well-separated pair decomposition of $P$ of size  $(1/\varepsilon)^{O(d)} n$ in $O(1)$ rounds, using $(1/\varepsilon)^{O(d)} n$ total space and $O(n^\delta)$ local space per machine for any $\delta \in (0,1)$.
\end{theorem}

\section{Applications of WSPD}\label{sec:appendix-apps}
In this section, we present four applications of our \textsf{WSPD}: fully scalable MPC algorithms for constructing a spanner, computing the closest pair, approximating the diameter, and constructing all $k$-nearest neighbors.
We begin by describing the first three applications, all of which are based on the algorithms in~\cite{har2011geometric}.
These applications rely on the representative points in every pair in a \textsf{WSPD}.
In contrast, our $k$-nearest neighbors algorithm follows the approach of Callahan and Kosaraju~\cite{callahan1995decomposition}, which relies on angular decomposition using cones, which is specific to Euclidean space.

\subsection{Spanner, Diameter Approximation, and Closest pair}\label{sec:appendix-three-apps}
In this subsection, we present $O(1)$-round MPC algorithms for constructing a $(1+\varepsilon)$-spanner, computing a $(1-\varepsilon)$-approximate diameter, and finding the closest pair for a set $P$ of $n$ points in $d$-dimensional Euclidean space or a doubling metric space with doubling dimension $d$ for a constant $d$.
Here, all algorithms use a \textsf{WSPD}, following the algorithms in~\cite{har2011geometric}.
Therefore, they use total space linear in the size of the \textsf{WSPD} used.
Recall that we can construct a $(1/\varepsilon)$-\textsf{WSPD} of size $(1/\varepsilon)^{O(d)}n$ and $(1/\varepsilon)^{O(\ddim)}n\log^2 n$ for $d$-dimensional Euclidean space and doubling metric space with doubling dimension $d$, respectively, for $\varepsilon>0$.
Our results in this section are obtained by combining Lemma~\ref{lem:ddim-wspd-const-prob}, Lemma~\ref{lem:ddim-wspd-high-prob}, and Theorem~\ref{thm:complete-wspd-euclid}.

\subparagraph*{Algorithms in~\cite{har2011geometric}.}
Given a set $P$ of $n$ points in $\mathbb{R}^d$, the three problems can be solved using a \textsf{WSPD} of $P$ by computing the distance between the representative points in each pair of a \textsf{WSPD}.
Specifically, a $(1+\varepsilon)$-spanner can be obtained by forming a graph whose vertex set consists of the representative points, with edges between the representative points of each pair of a $(c/\varepsilon)$-\textsf{WSPD}, where $c \geq 16$.
The number of edges equals the size of the \textsf{WSPD}.
To approximate the diameter within a factor of $(1 - \varepsilon)$, it suffices to take the maximum distance between the representative points of each pair of a $(4/\varepsilon)$-\textsf{WSPD}.
Similarly, the closest pair in $P$ is the pair of representative points with the minimum distance among all pairs of the $2$-\textsf{WSPD}. In particular, they showed that there exists a pair $(A,B)$ of the $2$-\textsf{WSPD} such that $|A|=|B|=1$, and
the pair of points in $A$ and $B$ is the closest pair in $P$.

All the approaches of~\cite{har2011geometric} rely on the triangle inequality and are therefore applicable to doubling metric spaces.

\subparagraph*{MPC implementation.}
To construct a $(1+\varepsilon)$-spanner, we first construct a $(c/\varepsilon)$-\textsf{WSPD} $\wspd$, where $c \geq 16$.
Then, we only need to compute the distance between the representative points of each pair of $\wspd$. 
This step requires only local computation, as all pairs in the $\wspd$ are explicitly stored together with their representative points.

\begin{theorem}\label{thm:euclid-spanner}
    Let $P$ be a set of $n$ points in $\mathbb R^d$ for constant $d$, and let $\varepsilon>0$.
    One can construct a $(1+\varepsilon)$-spanner of $P$ with $(1/\varepsilon)^{O(d)} n$ edges in $O(1)$ rounds, using $(1/\varepsilon)^{O(d)} n$ total space and $O(n^\delta)$ local space per machine for any $\delta \in (0,1)$.
\end{theorem}
\begin{theorem}
    Let $P$ be a set of $n$ points in a metric space of doubling dimension $\ddim$, for constant $\ddim$, and let $\varepsilon>0$.
    One can construct a $(1+\varepsilon)$-spanner of $P$ with $(1/\varepsilon)^{O(\ddim)} n\log^2 n$ edges in $O(1)$ rounds with a constant success probability, using $(1/\varepsilon)^{O(\ddim)} n\log^2 n$ total space and $O(n^\delta)$ local space per machine for any $\delta \in (0,1)$.
    The success probability can be amplified to $1-1/n^{\Omega(1)}$ by increasing the 
    total space by a factor of $O(\log n)$
\end{theorem}

For $(1-\varepsilon)$-approximate diameter, we need only one more step, which computes the maximum distance between the representative points in each pair of a \textsf{WSPD}.
More specifically, we first construct a $(4/\varepsilon)$-\textsf{WSPD} $\wspd$, and compute the distance between the representative points of each pair of $\wspd$ as for the previous problem.
After that, we compute the maximum of them in $O(1)$ rounds.

\begin{theorem}\label{thm:euclid-diameter}
    Let $P$ be a set of $n$ points in $\mathbb R^d$ for constant $d$, and let $\varepsilon>0$.
    One can compute $(1-\varepsilon)$-approximate diameter of $P$ in $O(1)$ rounds, using $(1/\varepsilon)^{O(d)} n$ total space and $O(n^\delta)$ local space per machine for any $\delta \in (0,1)$.
\end{theorem}
\begin{theorem}
 Let $P$ be a set of $n$ points in a  metric space of doubling dimension $\ddim$ for constant $\ddim$, and let $\varepsilon>0$.
 One can compute $(1-\varepsilon)$-approximate diameter of $P$ 
  in $O(1)$ rounds with a constant success probability, using $(1/\varepsilon)^{O(\ddim)} n\log^2 n$ total space and $O(n^\delta)$ local space per machine for any $\delta \in (0,1)$.
    The success probability can be amplified to $1-1/n^{\Omega(1)}$ by increasing the 
    total space by a factor of $O(\log n)$.
\end{theorem}

Similarly, we can find the (exact) closest pair.
We first construct a $2$-\textsf{WSPD} $\wspd$, and compute the distance between the representative points of each pair of $\wspd$ as for the previous problem.
After that, we compute the minimum of them in $O(1)$ rounds.

\begin{theorem}\label{thm:euclid-closest-pair}
    Let $P$ be a set of $n$ points in $\mathbb R^d$ for constant $d$, and let $\varepsilon>0$.
    One can compute the minimum distance between any two points in $P$ along with a realizing pair in $O(1)$ rounds, using $O(n)$ total space and $O(n^\delta)$ local space per machine for any $\delta \in (0,1)$.
\end{theorem}
\begin{theorem}
 Let $P$ be a set of $n$ points in a  metric space of doubling dimension $\ddim$ for constant $\ddim$.
  One can compute the minimum distance between any two points in $P$ along with a realizing pair in $O(1)$ rounds with a constant success probability, using $O(n\log^2 n)$ total space and $O(n^\delta)$ local space per machine for any $\delta \in (0,1)$.
    The success probability can be amplified to $1-1/n^{\Omega(1)}$ by increasing the total space by a factor of $O(\log n)$.
\end{theorem}

\subsection{$k$-Nearest Neighbors in Euclidean Space}\label{sec:appendix-k-nn}
In this subsection, we present our $O(k)$-round MPC algorithm for constructing all $k$-nearest neighbors for a set $P$ of $n$ points in $\mathbb{R}^d$ for constant $d$ using $O(kn)$ total space and $O(kn^\delta)$ local space per machine for any $\delta \in (0,1)$.
Our algorithm is based on the approach of Callahan and Kosaraju~\cite{callahan1995decomposition}.
The \emph{$k$-nearest neighbors} of a point $q$ in $P$ are the $k$ points in $P$ whose distances to $q$ are the smallest among all points in $P$. 
Now let $T$ be a compressed quadtree, and $\mathcal W$ be the 3-\textsf{WSPD}
constructed from $T$.

\medskip
We use two key geometric observations of~\cite{callahan1995decomposition}.  
First, for any pair $(A, B)$ of $\mathcal W$, if a point $a\in A$ is one of the $k$-nearest neighbors of a point $b\in B$, then $|B| \leq k$.
Next, let $U$ be a finite set of direction vectors in $\mathbb R^d$ such that every direction lies within an angle of at most $1 / 2$ radian from some vector in $U$ of size $O(2^d)=O(1)$ by~\cite{yao1982constructing}. 
For a point $p \in \mathbb{R}^d$ and a direction $\uVec \in U$, define $\cone(p, \uVec)$ as the cone with apex $p$, axis in direction $\uVec$, and opening angle at most $1/2$ radian.
For any two nodes $u,v$ in the compressed quadtree and direction $\uVec \in U$, if a point $p \in P \cap \nodeCell{u} \cap \cone(o_v, \uVec)$ is farther from $o_v$ than at least $k$ other points in $P \cap \nodeCell{u} \cap \cone(o_v, \uVec)$, then $\nodeCell{v}$ cannot contain any of the $k$-nearest neighbors of $p$, where $o_v$ denotes the center of $\nodeCell{v}$.
See Figure~\ref{fig:CK-obs2}.

\begin{figure}
    \centering
    \includegraphics[width=0.6\textwidth]{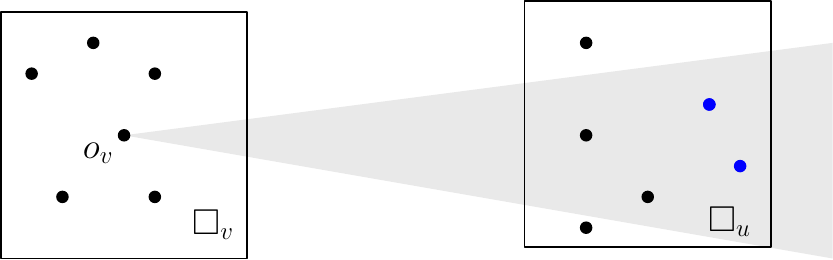}
        \caption{\small 
        The red cone is $\cone(o_v,\uVec)$ for a direction $\uVec\in U$.
        When $k=2$, any point in $\nodeCell{v}$ cannot be a $k$-nearest neighbor of the blue points.
    \label{fig:CK-obs2}}
\end{figure}

Motivated by these observations, we define the following structure as in~\cite{callahan1995decomposition}.
The \emph{search set}, denoted by $S(v,\uVec)$, of a node $v$ for a direction $\uVec$ is defined as the set of $k$ nearest points to $o_v$ among all points that lie in some cell $\cell\cap \cone(o_{v}, \uVec)$ paired with $\nodeCell{v}$ in $\wspd$, such that $|P\cap\cell| \leq k$.
Let $v_p$ be the leaf of $\qt$ corresponding to $p\in P$, and let $\primNode{p}$ be the lowest ancestor of $v_p$ in $\qt$ such that the union of the search sets for a direction $\uVec\in U$ of all nodes on the path from $v_p$ to $\primNode{p}$ in $\qt$ contains at least $k$ points.
By the previous two observation, we can prune the search space: 
Let $N(p, \uVec)$ denote the set of the $k$ nearest points of $p$ in the union of the search sets of all nodes on the path from $v_p$ up to the $(\lvl{v_p}-\lambda)$-level ancestor of $\primNode{p}$ in $\qt$,
where $\lambda$ is a constant specified~\cite{callahan1995decomposition}, and 
$\lvl{v_p}$ denotes the \emph{level} of $v_p$,
that is, the distance from $v_p$ to the root of $\qt$.
Then if $p$ is one of the $k$-nearest neighbors of some point $q$ in $P$, then
$q$ is contained in the union of $N(p, \uVec)$ over all $\uVec \in U$. 
If we have the union of $N(p,\uVec)$ for every point $p$, then we can find all $k$-nearest neighbors in $O(k)$ rounds using \textsf{Sorting} and \textsf{Minimum} operations. Thus the problem reduces to computing all $N(p,\uVec)$ for all points and all directions across the machines.

\medskip
We now describe our MPC algorithm for computing all $k$-nearest neighbors, following the algorithm of~\cite{callahan1995decomposition}.
More specifically, we find $N(p, \uVec)$ for all $p \in P$ and $\uVec \in U$, and then we compute the $k$-nearest neighbors of all points in parallel.
We first construct $S(v, \uVec)$ for every node $v$ in $\qt$ and every direction $\uVec \in U$.
To do this, for each node $v$ of $\qt$ such that $P\cap \nodeCell{v}\leq k$, we maintain all points in $\nodeCell{v}$, before constructing $\wspd$.
After that, we find $\primNode{p}$ for every point $p$.
Observe that there might be more than a constant number of nodes in the path in $\qt$ between $v_p$ and $\primNode{p}$ for a point $p$ and a direction $\uVec$, since the search sets of some nodes for a direction $\uVec$ might be empty.
Thus, it is unclear if one can compute $\primNode{p}$ in $O(1)$ rounds.
To avoid this case, we construct a new tree $\qt_\uVec$ consisting of the nodes of $\qt$ whose search set are non-empty.
We then find $\primNode{p}$ for all points $p$ in $P$ in parallel by propagating the search set of each node through its descendants in $\qt_\uVec$.
After that, we construct $N(p, \uVec)$ for every point $p \in P$ in a similar manner.

For a node $u$ of $T$, we use $\lvl{u}$ to denote the level of $u$, which is the distance from $u$ to the root of the tree, and $\ithDepParent{u}{i}$ to denote  the $i$-level ancestor of $u$ for $i>\lvl{u}$.

\subparagraph*{Maintaining at most $k$ points per cell.}
Suppose that we have all nodes of the compressed quadtree $\qt$ constructed by Lemma~\ref{lem:qt-cells}.
For each node $v\in \qt$, we associate $v$ with all points in $\nodeCell{v}$ if $|P \cap \cell| \leq k$.
To do this, we define two auxiliary points $x_v$ and $y_v$ representing the first and last vertices in $\nodeCell{v}$ under the $\mathcal{Z}$-order, respectively.
We sort all auxiliary points along with $P$ in the $\mathcal{Z}$-order.
In case of a tie between an auxiliary point and an identical point in $P$, the auxiliary point precedes.
Observe that the points contained in $\nodeCell{v}$ lie between $x_v$ and $y_v$ under the $\mathcal{Z}$-order.
See Figure~\ref{fig:knn-k-points}(a).
Therefore, we can determine whether $\nodeCell{v}$ contains at most $k$ points by comparing the ranks of $x_v$ and $y_v$ in the $\mathcal{Z}$-order.
If $|P \cap \nodeCell{v}| \leq k$, we associate $v$ with all points contained in $\nodeCell{v}$.

\begin{figure}
    \centering
    \includegraphics[width=0.65\textwidth]{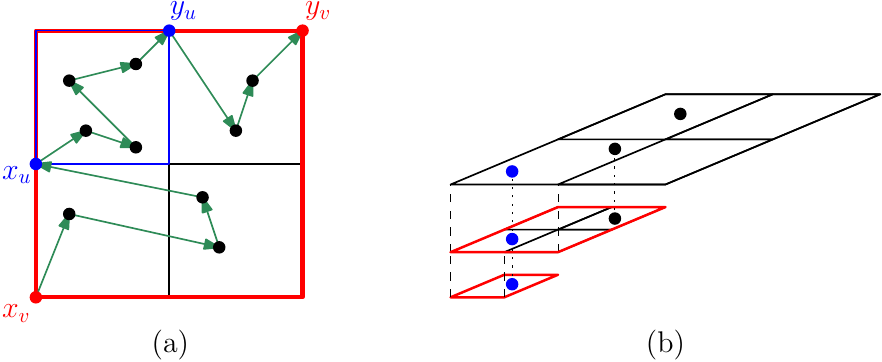}
        \caption{\small (a) The red cell is $\nodeCell{v}$ and the blue cell is $\nodeCell{u}$.
        The $\mathcal{Z}$ order is colored green.
        (b) When $k = 2$, two red cells are the only ones that contain at most $k$ points, including the blue point.
    \label{fig:knn-k-points}}
\end{figure}

\begin{lemma}
    One can associate each node in $\qt$ whose cell contains at most $k$ points with the corresponding subset of $P$ in $O(1)$ rounds, using $O(kn)$ total space and $O(kn^\delta)$ local space per machine.
\end{lemma}
\begin{proof}
    Each machine first constructs auxiliary points for all nodes in $\qt$ that it stores.
    Let $\corrM{v}$ be the machine that stores a node $v$ of $\qt$.
    Then, $\corrM{v}$ associates its auxiliary points $x_v$ and $y_v$ with the ID of $\corrM{v}$.
    Note that all auxiliary points for all nodes stored on a machine fit in that machine, and we can sort all points in $P$ along with their auxiliary points in $O(1)$ rounds.

    Next, for each node $v$, we determine whether $\nodeCell{v}$ contains at most $k$ points.
    Let $\corrM{x_v}$ and $\corrM{y_v}$ be the machines that store $x_v$ and $y_v$, respectively.
    We compute, for all nodes $v$, the successor of $x_v$ and the predecessor of $y_v$ among the points of $P$ in $O(1)$ rounds, using the \textsf{Predecessor} algorithm and a slight modification thereof.
    Consequently, the ranks of these points are stored in $\corrM{x_v}$ and $\corrM{y_v}$, respectively.
    These machines then send the rank information to $\corrM{v}$, which then computes $|P \cap \nodeCell{v}|$.

    Now, for each node $v$ such that $|P \cap \nodeCell{v}| \leq k$, we collect all points in $\nodeCell{v}$, and associate them with $v$.
    To do this, we send all points in $P \cap \nodeCell{v}$ to $\corrM{v}$.
    First, $\corrM{v}$ broadcasts $\nodeCell{v}$ and its ID to all machines between $\corrM{x_v}$ and $\corrM{y_v}$ in $O(1)$ rounds.
    Note that $\corrM{v}$ can obtain the IDs of $\corrM{x_v}$ and $\corrM{y_v}$ from the messages including the rank information of $x_v$ and $y_v$, provided that their machine IDs are also included.
    Each machine that stores a point in $P \cap \nodeCell{v}$ now knows $\nodeCell{v}$, and thus it can determine which point it stores lies in $\nodeCell{v}$.
    Such points are then sent to $\corrM{v}$, which associates $v$ with them.
    It takes $O(1)$ rounds, since each point in $P$ is contained in at most $k$ cells, each containing at most $k$ points from $P$.
    This follows from the fact that each node of the compressed quadtree splits its corresponding subset of $P$ into at least two non-empty parts.
    See Figure~\ref{fig:knn-k-points}(b).
\end{proof}

\subparagraph*{Constructing all search sets.}
We now construct the compressed quadtree $\qt$ by Lemma~\ref{lem:qt-association}, and construct a $3$-\textsf{WSPD} $\wspd$ from the quadtree $\qt$, following Lemma~\ref{lem:wspd-Euclidean}.
Here, we consider every pair in $\wspd$ as an ordered pair.
Let $(u, v)$ denote a pair of $\wspd$ consisting of the representations of two nodes $u, v \in \qt$.
After constructing $\wspd$, we duplicate each pair by adding its reverse: for every $(u, v)$, we generate $(v, u)$.

Next, we construct the search sets of all nodes in $\qt$. 
We begin by constructing a direction set $U$.  
For each pair $(u, v)$ with $|P \cap \nodeCell{u}| \leq k$, we generate a tuple $(v, \uVec, p)$ for every point $p \in P \cap \nodeCell{u}$ such that lies in the cone $\cone(o_v, \uVec)$ for some direction $\uVec \in U$.  
Once all such tuples are generated, we construct the search set $S(v, \uVec)$ for each node $v$ and direction $\uVec$ by taking the $k$ nearest points to $o_v$ among the points appearing in tuples of the form $(v, \uVec, \cdot)$.

\begin{lemma}
    All search sets for all nodes in $\qt$ and all directions in $U$ can be constructed in $O(k)$ rounds, using $O(kn)$ total space and $O(kn^\delta)$ local space per machine.
\end{lemma}

\begin{proof}
    By Lemma~\ref{lem:qt-association} and Lemma~\ref{lem:wspd-Euclidean}, it takes $O(1)$ rounds to construct $\qt$ and $\wspd$ using $O(n)$ total space and $O(n^\delta)$ local space per machine,  since $\varepsilon=3$.
    Recall that all pairs of $\wspd$ are distributed across machines.
    Accordingly, each machine duplicates every pair $(u, v)$ it stores.

    The first machine constructs the direction set $U$ and broadcasts it to all machines in $O(1)$ rounds, which can be done since $|U| =O(1)$.
    Then, each machine generates the tuples for the pairs $(u, v)$ it stores if $|P \cap \nodeCell{u}| \leq k$.
    Since the number of such tuples per pair is $O(k)$, the total number of tuples generated per machine fits within its memory.
    Furthermore, all tuples can be sorted in the lexicographic order of their nodes and directions in $O(1)$ rounds.

    After sorting, for each node $v$ and direction $\uVec \in U$, all tuples of the form $(v, \uVec, \cdot)$ are stored on a contiguous block of machines, denoted by $\setM_{v, \uVec}$.
    If $|\setM_{v, \uVec}| > 1$ and the first and last machines in this block also store tuples for other nodes or directions, then these machines send the tuples for $v$ and $\uVec$ to the next and previous machine, respectively, to ensure that $\setM_{v, \uVec}$ contains only tuples for $v$ and $\uVec$.
    This arrangement takes $O(1)$ rounds.

    Once this is done, each $\setM_{v, \uVec}$ computes the $k$ nearest points to $o_v$ among the tuples it stores.
    This is achieved by iteratively selecting the nearest point to $o_v$ up to $k$ times using the \textsf{Minimum} algorithm, which requires $O(k)$ rounds in total.
\end{proof}

\subparagraph*{Constructing $\qt_\uVec$.}
We now construct a new tree $\qt_\uVec$ for each direction $\uVec \in U$, consisting of the nodes of $\qt$ whose search sets are non-empty.  
To perform this in parallel for all directions, we first arrange the search sets for each $\uVec$ on a contiguous block of machines, ensuring that blocks for different directions are disjoint.  
This is achieved by sorting all search sets in the lexicographic order of their directions.  
Once sorted, we proceed to construct $\qt_\uVec$ for each $\uVec \in U$.

Before sorting, however, we ensure that $\qt_\uVec$ for every direction includes all leaves of $\qt$, even if the search set for a leaf of $\qt$ is empty.
This is necessary to correctly construct $N(p, \uVec)$ for all $p \in P$ at a later step.  
Note that for a leaf $v_p \in \qt$ corresponding to a point $p \in P$, the search set $S(v_p, \uVec)$ for a direction $\uVec$ may be empty and thus may not have been generated in the previous step.  
In such cases, we explicitly generate $S(v_p, \uVec)$ as an empty set.

\begin{lemma}\label{lem:knn-Tu}
    We can construct all trees $\qt_\uVec$ for $\uVec \in U$ in $O(1)$ rounds, using $O(kn)$ total space and $O(kn^\delta)$ local space per machine.
\end{lemma}
\begin{proof}
    It suffices to show that we can arrange the search sets for each direction on contiguous machine blocks, while also generating any missing empty search sets, within $O(1)$ rounds.
    Let $\setM_\uVec$ denote the block of machines that stores only the search sets for direction $\uVec$.
    We can then construct $\qt_\uVec$ on $\setM_\uVec$ by Lemma~\ref{lem:qt-association}, and thus construct all trees for all directions in $U$ in parallel within $O(1)$ rounds.

    We begin by generating an empty search set for every point in $P$ and every direction in $U$. 
    To do this, each machine generates an empty search set of the leaves in $\qt$ for all directions of $U$, which it stores.
    Since $|U|=O(1)$, each machine generates $O(1)$ empty search sets per point it stores, which fit within its local memory.
    Therefore, we can sort all search sets along with the generated empty ones in the lexicographic order of their directions in $O(1)$ rounds.
    We can then remove the empty copies in $O(1)$ rounds, since all duplicates for a fixed point $p$ and direction $\uVec$ are stored consecutively.
 
    Next, we arrange the search sets.
    For each direction $\uVec\in U$, if $|\setM_{\uVec}| > 1$ and the first and last machines in $\setM_\uVec$ also store the search sets for other directions, then these machines send the search sets for $\uVec$ to the next and previous machine, respectively, to ensure that $\setM_{\uVec}$ contains only tuples for $v$ and $\uVec$.
    This final rearrangement takes $O(1)$ rounds, and it concludes this proof.
\end{proof}

\subparagraph*{Finding $\primNode{p}$} Now, we find all $\primNode{p}$ for all points $p$ and directions $\uVec\in U$ in parallel.
Moreover, for all $p$, we will construct the union of the search sets for $\uVec$ of the nodes on the path from $v_p$ to $\primNode{p}$.
Since $\qt_\uVec$ preserves the hierarchical structure of the original compressed quadtree, the node $\primNode{p}$ in $\qt_\uVec$ is the same as that in $\qt$.

As aforementioned, we find $\primNode{p}$ by propagating the search sets of each node through its descendants in $\qt_\uVec$.
More specifically, each node $v$ in $\qt_\uVec$ first sends its search set $S(v,\uVec)$ along with itself to its children, and then the children of $v$ construct the union of its search set and $S(v,\uVec)$.
After that, $v$ repeatedly sends the most recently received set and node to its children, while maintaining the union of the received sets.
If the size of the union exceeds $k$ during propagation, $\parent{v}$ stops sending search sets to $v$.
By the following lemma, $v_p$ then has $\primNode{p}$ and the union of the search sets for $\uVec$ of the nodes on the path from $v_p$ to $\primNode{p}$.

\begin{figure}
    \centering
    \includegraphics[width=0.7\textwidth]{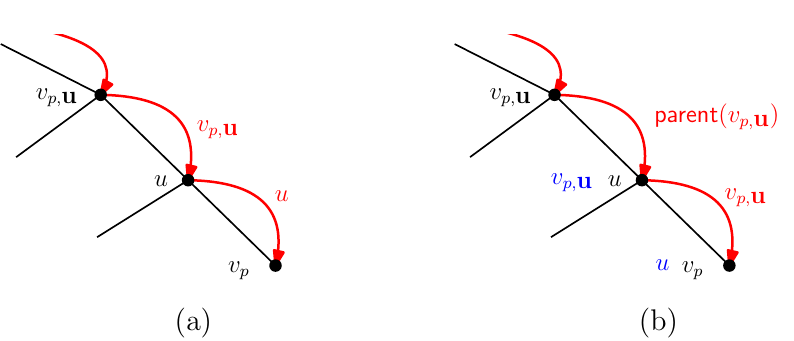}
        \caption{\small (a) Illustration of the first step of the propagation.
            The red arrows represent which node is propagated.
            (b) Illustration of the second step of the propagation.
            Each node maintains the left node colored blue.
    \label{fig:propagation}}
\end{figure}

\begin{lemma}
    After propagation, for all points $p$, $v_p$ has $\primNode{p}$ and the union of the search sets for $\uVec$ of the nodes on the path from $v_p$ to $\primNode{p}$.
\end{lemma}
\begin{proof}
    It suffices to show that $\primNode{p}$ reaches $v_p$ during the propagation.
    Assume that $\primNode{p}$ does not reach $v_p$.
    Then there exists a node $v$ on the path from $\primNode{p}$ to $v_p$ in $\qt_\uVec$ that constructs a union of search sets of size at least $k$, thereby preventing the propagation from $\primNode{p}$.
    Even after that, $v$ continues to propagate the most recently received search sets.
    Thus, it must be the case that some node $v' \ne \primNode{p}$ on the path between $\primNode{p}$ and $v$ causes $v$ to stop receiving sets from its parent.
    This contradicts the definition of $\primNode{p}$.
\end{proof}

Next, we determine which node is $\primNode{p}$ for each point in $P$ by checking which node is the most recently received node.
Furthermore, each leaf $v_p$ has the union of the search sets of the nodes in the path from $v_p$ up to $\primNode{p}$.

\begin{lemma}\label{lem:knn-primeNode}
    For all points $p$ in $P$ and directions $\uVec \in U$, we can determine $\primNode{p}$ in $O(k)$ rounds, using $O(kn)$ total space and $O(kn^\delta)$ local space per machine.
    Moreover, each machine storing $v_p$ of $\qt_\uVec$ also stores $\primNode{p}$ and the union of the search sets of the nodes on the path from $v_p$ up to $\primNode{p}$.
\end{lemma}
\begin{proof}
    Since each tree $\qt_\uVec$ is stored on a contiguous block of machines that is disjoint from the blocks for other directions by Lemma~\ref{lem:knn-Tu}, this can be done in parallel for all directions in $U$.
    Therefore, it suffices to consider a fixed direction $\uVec \in U$ and show that we can compute $\primNode{p}$ for all $p \in P$ in $O(k)$ rounds.

    Let $\corrM{v}$ denote the machine storing node $v$ in $\qt_\uVec$.  
    Initially, for every node $v$, $\corrM{v}$ sends $v$ itself and $S(v, \uVec)$ to the machines storing its children. 
    For each child $v'$, $\corrM{v'}$ then construct $S(v', \uVec) \cup S(v, \uVec)$.  
    If $\corrM{v}$ subsequently receives $S(\parent{v}, \uVec)$, it also updates its union accordingly. After that, each machine continues propagating the new received set and node to its children.
    We repeat this, and each machine maintains the union of all received search sets.
    During the propagation, if the size of the union for $v$ becomes at least $k$, $\corrM{v}$ sends a stop signal to $\corrM{\parent{v}}$, which then stops propagating further search sets to $v$.
    
    Observe that at the end, each machine storing a leaf $v_p$ also obtains $\primNode{p}$ and the union of the search sets of the nodes on the path from $v_p$ to $\primNode{p}$.  

    We now show that it takes $O(k)$ rounds.
    By Lemma~\ref{lem:qt-association}, each node in $\qt_\uVec$ knows the IDs of the machines storing its parent and children, which allows search sets to be propagated between connected nodes.  
    Moreover, in each round, every machine sends at most $2^d n^\delta$ messages each of size $k$, so a single propagation step can be completed in one round.  
    Since the union size increases by at least one in each iteration, and all internal nodes in $\qt_\uVec$ have non-empty search sets, the total number of iterations is at most $k$.
\end{proof}

\subparagraph*{Constructing $N(p, \uVec)$.}
We now construct $N(p, \uVec)$ for every point $p \in P$, following a similar approach as in the previous step.
More specifically, we propagate search sets for $O(\lambda)$ rounds.
Then, for each point $p$, the node $\primNode{p}$ obtains the union of the search sets of the nodes on the path from $\primNode{p}$ to $\ithDepParent{\primNode{p}}{\lvl{\primNode{p}} - \lambda}$.
After that, we send the union of search sets of $\primNode{p}$ to $v_p$ in $O(k)$ rounds by propagation as in the previous step.

\begin{lemma}\label{lem:knn-N}
    For all points in $P$ and direction $\uVec\in U$, we can construct $N(p,\uVec)$ in $O(k)$ rounds, using $O(kn)$ total space and $O(k n^\delta)$ local space per machine.
    Moreover, each machine storing $v_p$ of $\qt_\uVec$ also stores $N(p,\uVec)$.
\end{lemma}
\begin{proof}
    Let $\lambda$ be the constant specified by~\cite{callahan1995decomposition}.
    For every node $v$ in $\qt_\uVec$, let $\tilde N(v, \uVec)$ denote the union of the search sets of the nodes on the path from $v$ to its $(\lvl{v} - \lambda)$-level ancestor.
    Following the approach of the proof of Lemma~\ref{lem:knn-primeNode}, we can determine which node is $\primNode{p}$ and construct $\tilde N(\primNode{p}, \uVec)$ for all nodes $v$ in $O(\lambda)$ rounds.
    Notice that $\tilde N(\primNode{p}, \uVec)$ is stored in $\corrM{\primNode{p}}$.

    Then, for all points $p$ in $P$, we send $\tilde N(\primNode{p},\uVec)$ to the machine storing $v_p$ by the same propagation manner.
    This can be done in $O(k)$ round since $|\tilde N(\primNode{p}, \uVec)|=O(k)$ and the distance between $\primNode{p}$ and $v_p$ is at most $k$.
    Subsequently, each machine $\corrM{\primNode{p}}$ constructs $N(p, \uVec)$ by combining $\tilde N(\primNode{p}, \uVec)$ with the received set.
\end{proof}

\subparagraph*{Constructing $k$ nearest neighbors.}
Now, we have $N(p, \uVec)$ for every $p \in P$ and every direction $\uVec \in U$.  
For each point $p$, we associate it with every point $q \in N(p, \uVec)$, and sort the resulting pairs in the lexicographic order of $q$.
As a result, all candidates for the $k$ nearest neighbors of each point $q$ are stored in a contiguous block of machines, denoted by $\setM_q$.
We then compute the $k$ nearest neighbors of every point $q$ by repeatedly selecting the nearest point from $\setM_q$, $k$ times.

\begin{theorem}\label{thm:euclid-KNN}
    Given a set $P$ of $n$ points in $\mathbb R^d$ for constant $d$, one can compute $k$ nearest neighbors of all points of $P$ in $O(k)$ rounds, using $O(kn)$ total space and $O(kn^\delta)$ local space per machine for any $\delta \in (0,1)$.
    For a given point $p\in P$, finding $k$-nearest neighbors of $p$ takes $O(1)$ rounds.
\end{theorem}
\begin{proof}

    It is straightforward to compute the $k$ nearest neighbors for all points in $P$ in $O(k)$ rounds, as both sorting and minimum computations can be performed in $O(1)$ rounds.

    We now show that, given a query point $q \in P$, we can return its $k$ nearest neighbors in $O(1)$ rounds.
    As is standard, we assume that there is a designated machine $M$ that receives the query.
    First, $M$ broadcasts $q$ to all machines in $O(1)$ rounds.
    Subsequently, each machine checks whether it stores any of the $k$ nearest neighbors of $q$, and if so, sends them to $M$, which then returns the $k$ nearest neighbors of $q$.
\end{proof}
\bibliography{references}


\end{document}